\documentclass{article}
\usepackage{fullpage}
\usepackage{amsmath,amsfonts,amsthm,mathtools}
\usepackage{algorithm}
\usepackage{url}  
\usepackage{color}
\usepackage{tikz}
\usepackage{graphicx}  
\usepackage{algorithm}
\usepackage{algorithmic}
\usepackage[sort]{cite}
\usepackage{amssymb} 
\usepackage{hyperref}

\newcommand{\citep}{\cite}

\usepackage{comment}

\definecolor{darkgreen}{rgb}{0,0.5,0}
\definecolor{darkred}{rgb}{0.7,0,0}
\definecolor{purple}{rgb}{0.7,0,0.7}
\definecolor{teal}{rgb}{0.3,0.8,0.8}
\newcommand{\kibitz}[2]{\ifnum\Comments=1\textcolor{#1}{#2}\fi}

\usepackage{prettyref}
\newcommand{\pref}[1]{\prettyref{#1}}

\newcommand{\savehyperref}[2]{\texorpdfstring{\hyperref[#1]{#2}}{#2}}
\newrefformat{eq}{\savehyperref{#1}{\textup{(\ref*{#1})}}}
\newrefformat{eqn}{\savehyperref{#1}{Equation~\ref*{#1}}}
\newrefformat{lem}{\savehyperref{#1}{Lemma~\ref*{#1}}}
\newrefformat{def}{\savehyperref{#1}{Definition~\ref*{#1}}}
\newrefformat{line}{\savehyperref{#1}{line~\ref*{#1}}}
\newrefformat{thm}{\savehyperref{#1}{Theorem~\ref*{#1}}}
\newrefformat{corr}{\savehyperref{#1}{Corollary~\ref*{#1}}}
\newrefformat{sec}{\savehyperref{#1}{Section~\ref*{#1}}}
\newrefformat{ssec}{\savehyperref{#1}{Subsection~\ref*{#1}}}
\newrefformat{app}{\savehyperref{#1}{Appendix~\ref*{#1}}}
\newrefformat{ass}{\savehyperref{#1}{Assumption~\ref*{#1}}}
\newrefformat{ex}{\savehyperref{#1}{Example~\ref*{#1}}}
\newrefformat{fig}{\savehyperref{#1}{Figure~\ref*{#1}}}
\newrefformat{alg}{\savehyperref{#1}{Algorithm~\ref*{#1}}}
\newrefformat{rem}{\savehyperref{#1}{Remark~\ref*{#1}}}
\newrefformat{conj}{\savehyperref{#1}{Conjecture~\ref*{#1}}}
\newrefformat{prop}{\savehyperref{#1}{Proposition~\ref*{#1}}}
\newrefformat{proto}{\savehyperref{#1}{Protocol~\ref*{#1}}}
\newrefformat{prob}{\savehyperref{#1}{Problem~\ref*{#1}}}
\newrefformat{claim}{\savehyperref{#1}{Claim~\ref*{#1}}}
\newrefformat{fact}{\savehyperref{#1}{Fact~\ref*{#1}}}

\newtheorem*{theorem*}{Theorem}
\newtheorem*{lemma*}{Lemma}
\newtheorem{claim}{Claim}

\newcommand{\beq}{\begin{align}}
\newcommand{\eeq}{\end{align}}

\newcommand{\bin}{\textup{Bin}}

\newfont{\bbb}{msbm10 scaled 500}

\newfont{\bb}{msbm10 scaled 1100}

\newcommand{\EE}{\mbox{\bb E}}

\newcommand{\0}{{\texttt{0}}}
\newcommand{\1}{{\texttt{1}}}

\newcommand{\Mc}{{\cal M}}

\newcommand{\Tc}{{\cal T}}

\DeclareMathOperator{\Var}{var}

\DeclareMathOperator{\poly}{poly}

\newcommand{\avg}{{\mathbb E}}

\newcommand{\len}{\mathrm{len}}

\newtheorem{theorem}{Theorem}
\newtheorem{definition}{Definition}
\newtheorem{lemma}[theorem]{Lemma}
\newtheorem{proposition}[theorem]{Proposition}
\newtheorem{corollary}[theorem]{Corollary}

\newtheorem{remark}{\indent Remark}

\theoremstyle{definition}

\newcommand{\expec}[2][]{{\mathbb E}_{#1} \left [ #2 \right ] }

\newcommand{\prob}[2][]{\Pr_{#1} \left [ #2 \right ] }

\DeclareMathAlphabet{\mathpzc}{OT1}{pzc}{m}{it}

\DeclareMathOperator*{\argmax}{arg\,max}
\DeclareMathOperator*{\argmin}{arg\,min}

\begin{document}
\title{Trace Reconstruction: Generalized and Parameterized
}
\author{Akshay Krishnamurthy \and Arya Mazumdar \and Andrew McGregor \and Soumyabrata Pal
}

\maketitle
{\renewcommand{\thefootnote}{}\footnotetext{
Akshay Krishnamurthy is with Microsoft Research, New York. Arya Mazumdar is with University of California, San Diego. Andrew McGregor and Soumyabrata Pall are with College of Information and Computer Sciences, University of Massachusetts, Amherst. Emails: \texttt{\{akshay,arya,mcgregor,spal\}@cs.umass.edu}.  This work was supported in part by the National Science Foundation under  CCF1642658, 1637536, 1763618, 1934846, 1909046  and 1908849. Part of this work was presented in the European Symposium of Algorithms, 2019.
}
\renewcommand{\thefootnote}{\arabic{footnote}}
\setcounter{footnote}{0}

\begin{abstract}
In the beautifully simple-to-state problem of trace reconstruction, the goal is to
reconstruct an unknown  binary string $x$ given random ``traces'' of $x$ where each trace is generated by deleting each coordinate of $x$ independently with probability $p<1$. The problem is well studied both when the unknown string is arbitrary and when it is chosen uniformly at random. 
For both settings, there is still an exponential gap between upper and lower sample complexity bounds and our understanding of the problem is still surprisingly limited. In this paper, we consider natural parameterizations and generalizations of this problem in an effort to attain a deeper and more comprehensive understanding. Perhaps our most surprising results are:
\begin{enumerate}
\item We prove that $\exp(O(n^{1/4} \sqrt{\log n}))$ traces suffice for reconstructing arbitrary matrices. In the matrix version of the problem, each row and column of an unknown $\sqrt{n}\times \sqrt{n}$ matrix is deleted independently with probability $p$. Our results contrasts with the best known results for sequence reconstruction where the best known upper bound is $\exp(O(n^{1/3}))$. 
\item An optimal result for random matrix reconstruction: we show that $\Theta(\log n)$ traces are necessary and sufficient. This is in contrast to the problem for random sequences where there is a super-logarithmic lower bound and the best known upper bound is $\exp({O}(\log^{1/3} n))$.
\item We show that $\exp(O(k^{1/3}\log^{2/3} n))$ traces suffice to reconstruct $k$-sparse strings, providing an improvement over the best known sequence reconstruction results when $k = o(n/\log^2 n)$. 
\item We show that $\textrm{poly}(n)$ traces suffice if $x$ is $k$-sparse and we additionally have a ``separation" promise, specifically that the indices of  \1's in $x$ all differ by  $\Omega(k \log n)$. 
\end{enumerate}
\end{abstract}

\section{Introduction}
V. Levenshtein in \cite{levenshtein1997reconstruction} asked the following combinatorial question regarding reconstruction of  a sequence from its subsequences: how many subsequences of a particular length are necessary and sufficient to reconstruct the original sequence? He followed up with \cite{levenshtein2001efficient} and  \cite{levenshtein2001efficient2}  where upper and lower bounds were provided for different variations on the problem, along with efficient reconstruction algorithms. A similar question was studied in \cite{KRASIKOV1997344}: to find the minimum value of $t$ such that we can reconstruct any binary sequence provided we are given all subsequences of length $t$. In his paper \cite{levenshtein2001efficient}, Levenshtein also introduced the probabilistic version of the problem for discrete memoryless channels, stopping just short of introducing the trace reconstruction problem.

In the trace reconstruction problem, first proposed
by Batu et al.~\cite{BatuKKM04}, the goal is to reconstruct an unknown string
$x\in \{\0,\1\}^n$ given a set of random subsequences of $x$. Each
subsequence, or ``trace", is generated by passing $x$ through the
\emph{deletion channel} in which each entry of $x$ is deleted
independently with probability $p$. The locations of the deletions are
not known; if they were, the channel would be an \emph{erasure
  channel}. The central question is to find how many
traces are required to exactly reconstruct $x$ with high
probability.

This intriguing problem has attracted significant attention from a large
number of
researchers~\citep{KM05,ViswanathanS08,BatuKKM04,HolensteinMPW08,HoldenPP18,PeresZ17,HartungHP18,NazarovP17,DeOS17,McGregorPV14,2019arXiv190205101D,2019arXiv190309992C}.
In a recent breakthrough, De et al.~\cite{DeOS17} and Nazarov and Peres~\cite{NazarovP17}
independently showed that $\exp({O}((n/q)^{1/3}))$ traces suffice
where $q=1-p$. 
This bound is achieved by a \emph{mean-based} algorithm, which means
that the only information used is the fraction of traces that have a
\1 in each position.
While $\exp({O}((n/q)^{1/3}))$ is known to be optimal amongst
mean-based algorithms, the best algorithm-independent lower bound is
the much weaker $\Omega(n^{5/4}/\log n)$~\citep{holden2018lower}.

Many variants of the problem have also been considered including: (1)
larger alphabets and (2) an average case analysis where $x$ is drawn
uniformly from $\{\0,\1\}^n$.  Larger alphabets are only easier than
the binary case, since we can encode the alphabet in binary, e.g., by
mapping a single {\color{black}particular} character to $\1$ and the rest to $\0$. {\color{black} We can then solve the binary problem} and subsequently, repeat the process for all characters {\color{black}to reconstruct the entire string}.  In the average case analysis, the
state-of-the-art result is that $\exp({O}(\log^{1/3}(n)))$ traces
suffice\footnote{$p$ is assumed to be constant in that work.}, whereas
$\Omega(\log^{9/4} n / \sqrt{\log\log n})$ traces are
necessary~\citep{HartungHP18,HoldenPP18,holden2018lower}. Very recently, and concurrent with our work, other variants have been studied including a) where the bits of $x$ are associated with nodes of a tree whose topology determines the distribution of traces generated \cite{2019arXiv190205101D} and b) where $x$ is a codeword from a code with $o(n)$ redundancy \cite{2019arXiv190309992C}.


In order to develop a deeper understanding of this
intriguing problem, we consider fine-grained parameterization and
structured generalizations of trace reconstruction. We prove several
new results for these variations that shed new light on the
problem. Moreover, in studying these settings, we refine existing
tools and introduce new techniques that we believe may be helpful in
closing the gaps in the fully general problem.

\subsection{Our Results}
{\color{black}In all our results below, we have used the term \textit{with high probability} to imply that the statement holds with probability at least $1-o(1)$ where $o(1)$ is a term asymptotically going to 0 as the size of the input (typically the length of the string, $n$) grows.}

\subsubsection{Parametrizations}
We begin by considering parameterizations of the trace reconstruction
problem. Given the important role that sparsity plays in other
reconstruction problems (see, e.g., Gilbert and Indyk~\cite{GilbertI10}), we first
study the recovery of sparse strings. Here we prove the following
result.

\begin{theorem}
\label{thm:sparsity_intro}
Let $q\equiv 1-p$ be the retention
probability and assume that $q=\Omega(k^{-1/2} \log^{1/2} n)$. If $x\in \{0,1\}^n$ has at most $k$ non-zeros, $\exp(O((k/q)^{1/3}\log^{2/3} n))$ traces suffice to recover $x$
exactly, with high probability.
\end{theorem}

As some points of comparison, note that there is a trivial
$\exp(O(k/q+\log n))$ upper bound, which our result improves on with a
polynomially better dependence on $k/q$ in the exponent. {\color{black}The trivial bound is obtained by getting enough samples so that it is possible to obtain ${\poly}(n)$ samples where none of the \1s are deleted.} The best known result for the general case is
$\exp(O((n/q)^{1/3}))$~\citep{NazarovP17,DeOS17} and our result
is a strict improvement when $k=o(n/\log^{2} n)$. Note that since we have no restrictions on $k$ in the statement, improving upon $\exp(O((k/q)^{1/3}))$ would imply an
improved bound in the general setting.

Somewhat surprisingly, our actual result is considerably stronger
(See~\pref{corr:del} for a precise statement). We also obtain
$\exp(O((k/q)^{1/3}\log^{2/3} n))$ sample complexity in an asymmetric
deletion channel, where each \0 is deleted with probability
{\color{black}extremely} close to $1$, but each \1 is deleted with 
probability $p=1-q$. With such a channel, all but a vanishingly small
fraction of the traces contain only \1s, yet we are still able to
exactly identify the location of every \0. Since we can accommodate
$k=\Theta(n)$ this result also applies to the general case with an
asymmetric channel, yielding improvements over De et al.~\cite{DeOS17} and
Nazarov and Peres~\cite{NazarovP17}.

We elaborate more on our techniques in the next section, but the
result is obtained by establishing a connection between trace
reconstruction and learning binomial mixtures. There is a large body of work devoted to learning mixtures \cite{dasgupta1999learning,achlioptas2005spectral,kalai2010efficiently,belkin2010polynomial,arora2001learning,moitra2010settling,feldman2008learning,chan2013learning,hopkins2018mixture,hardt2015tight} where it is common to assume that the mixture components are
well-separated. In our context, separation corresponds to a promise that each pair of \1s in the original string is separated by a \0-run of a certain length.
Our second result concerns strings with a separation promise.
\begin{theorem}
\label{thm:gaps_intro}
If $x$ has at most $k$ \1s and each \1 is separated by \0-run of length $\Omega(k\log n)$, then, for any constant deletion probability $p$, ${\poly}(n)$ traces suffice to recover $x$ with high probability.
\end{theorem}

Note that reconstruction with $\textrm{poly}(n)$ traces is straightforward if every \1 is separated by a 
\0-run of length $\Omega(\sqrt{n \log n})$; the basic idea is that we can identify which \1s in a collection of traces correspond to the same \1 in the original sequence and then we can use the indices of these  \1s in their respective traces to infer the index of the \1 in the original string. However, reducing to $\Omega(k\log n)$ separation is rather involved and is perhaps the most technically challenging result in this paper.

Here as well, we actually obtain a slightly stronger result. Instead
of parameterizing by the sparsity and the separation, we instead
parameterize by the number of runs, and the run lengths, where a run
is a contiguous sequence of the same character. We require that each
\0-run has length $\Omega(r\log n)$, where $r$ is the total number
of runs.
 Note that this parameterization yields a stronger result since $r$ is at most $2k+1$ if the
 string is $k$ sparse, but it can be much smaller, for example if the
 \1-runs are very long. On the other hand, the best lower bound, which
is $\Omega(n^{5/4}/\log n)$~\citep{holden2018lower}, considers strings with
$\Omega(n)$ runs and run length $O(1)$.
 
Using the general approach used to prove Theorem \ref{thm:gaps_intro}, we can also prove an average case reconstruction result for sparse strings:  $\poly(n)$ traces suffice if each $x_i\sim \textrm{Ber}(\eta)$ where $\eta\leq c/\sqrt{n\log n}$ for some sufficiently small $c$. As mentioned, above if $\eta=1/2$, it was already known that a sub-polynomial number of traces sufficed for reconstruction. However, for random strings sparsity is not necessarily helpful. In fact, if $\eta=1/n$ it is relatively straightforward to argue that $\poly(n)$ traces are necessary since  with constant probability  $x$ has the form 
\[x=\underbrace{00\ldots 00}_{\geq n/4} 1\underbrace{00\ldots 00}_{\geq n/4}\] 
and identifying the position of the 1 requires  $\Omega(n)$ traces.

As our last parametrization, we consider a sparse testing problem. We
specifically consider testing whether the true string is $x$ or $y$,
with the promise that the Hamming distance between $x$ and $y$, $\Delta(x,y)$, is at most $2k$. This question is naturally related to sparse
reconstruction, since the difference sequence $x-y\in \{-1,0,1\}^n$ is $2k$ sparse, although of
course neither string may be sparse on its own. Here we obtain the
following result.

\begin{theorem}
\label{thm:hamming_intro}
For any pair $x,y \in \{\0,\1\}^n$ with $\Delta(x,y) \leq 2k$,
$\exp(O(k \log n))$ traces {\color{black}from the deletion channel with $p \le 1-k/n$} suffice to distinguish between $x$ and
$y$ with high probability.
\end{theorem}

\subsubsection{Generalizations}
Turning to generalizations, we consider a natural multivariate version of
the trace reconstruction problem, which we call \emph{matrix
  reconstruction}. Here we receive matrix traces of an unknown binary
matrix $X \in \{0,1\}^{\sqrt{n}\times\sqrt{n}}$, where each matrix
trace is obtained by deleting each row and each column with
probability $p$, independently. Here the deletion channel is much more
structured, as there are only $2\sqrt{n}$ random bits, rather than
$n$ in the sequence case. Our results show that we can exploit this
structure to obtain improved sample complexity guarantees.

In the worst case, we prove the following theorem.
\begin{theorem}
\label{thm:matrix_intro}
For the matrix deletion channel with deletion probability $p$,
\[\exp(O(n^{1/4}\sqrt{p\log n}/q))\] traces suffice to recover an
arbitrary matrix $X \in \{0,1\}^{\sqrt{n}\times\sqrt{n}}$ {\color{black}with high probability}.
\end{theorem}

While no existing results are directly comparable, it is possible to
obtain $\exp(O(n^{1/3}\log n))$ sample complexity via a combinatorial
result due to K\'os et al.~\cite{KosLS09}. This agrees with the results from the
sequence case, but is obtained using very different techniques.
Additionally, our proof is constructive, and the
algorithm is actually mean-based, so the only information it requires
are estimates of the probabilities that each received entry is \1. As
we mentioned, for the sequence case, both Nazarov and Peres~\cite{NazarovP17}
and De et al.~\cite{DeOS17} prove a $\exp(\Omega(n^{1/3}))$ lower bound for
mean-based algorithms. Thus, our result provides a strict separation
between matrix and sequence reconstruction, at least from the
perspective of mean-based approaches. 

Lastly, we consider the random matrix case, where every entry of $X$
is drawn iid from $\textrm{Ber}(1/2)$. Here we show that $O(\log n)$
traces are sufficient.
\begin{theorem}
\label{thm:random_matrix_intro}
For any constant deletion probability $p<1$, $O(\log n)$ traces
suffice to reconstruct a random $X\in
\{0,1\}^{\sqrt{n}\times\sqrt{n}}$ with high probability over the
randomness in $X$ and the channel.
\end{theorem}
This result is optimal, since with $o(\log n)$ traces, there is
reasonable probability that a row/column will be deleted from all
traces, at which point recovering this row/column is impossible.  The
result should be contrasted with the analogous results in the sequence
case. For sequences, the best results for random strings are
$\exp(O(\log^{1/3} n))$~\citep{HoldenPP18} and $\Omega(\log^{9/4}
n/\sqrt{\log \log n})$~\citep{holden2018lower}. In light of the lower
bound for sequences, it is perhaps surprising that matrix reconstruction admits
$O(\log n)$ sample complexity. 

In~\pref{sec:tensors}, we show that it is possible to extend both matrix reconstruction results to tensors in a reasonably straightforward way.

\subsection{Our Techniques}
To prove our results, we introduce several new techniques in addition to refining and extending  many existing ideas in prior trace reconstruction results.

 \pref{thm:sparsity_intro} is proved via a reduction from trace reconstruction to learning the parameters of a mixture of binomial distributions. Surprisingly, this natural connection does not seem to have been observed in the earlier literature. We then use a generalization of a complex-analytic approach introduced by De et al.~\cite{DeOS17} and Nazarov and Peres~\cite{NazarovP17} to prove a bound on the sample complexity of learning a binomial mixture. This generalization is to move beyond the analysis of Littlewood polynomials, i.e., polynomials with $\{-1,0,1\}$ coefficients, to the case where  coefficients have bounded precision. The generalization is not difficult.
  This is our simplest result to prove but we consider the final result to be revealing as it shows that sparsity plays a more important role than length in the complexity of trace reconstruction.

Our most technically involved result is \pref{thm:gaps_intro}. This is proved via an algorithm that constructs
a hierarchical clustering of the individual \1s in all received
traces according to their corresponding position in the original
string. This clustering step requires a careful recursion, where in
each step we ensure no false negatives (two \1s from the same origin
are always clustered together) but we have many false positives, which we
successively reduce. At the bottom of the recursion, we can identify a
large fraction of \1s from each \1 in the original string. However, as
the recursion eliminates many of the \1s, simply averaging the
positions of the surviving fraction leads to a biased estimate.  To
resolve this, we introduce a de-biasing step which eliminates even
more \1s, but ensures the survivors are unbiased, so that we can
accurately estimate the location of each \1 in the original string.
The initial recursion has $L = \log \log n$
levels, which is critical since the debiasing step involves
conditioning on the presence of $2^L$ \1s in a trace, which only
happens with probability $2^{-2^L} = {1}/{n}$. 

 \pref{thm:hamming_intro} leverages combinatorial arguments
about $k$-decks (the multiset of subsequences of a string) due
to Krasikov and Roditty~\cite{KRASIKOV1997344}.  
The result
demonstrates the utility of these combinatorial tools in trace
reconstruction. As further evidence for the utility of combinatorial
tools, the connection to $k$-decks was also used
by Ban et al.~\cite{ban2019beyond} in independent concurrent work on the
deletion channel.

 For~\pref{thm:matrix_intro}, we return to the complex-analytic approach and extend the Littlewood polynomial argument to
multivariate polynomials. Since the unknown matrices are $\sqrt{n}\times\sqrt{n}$, we can use a natural bivariate polynomial of degree $O(\sqrt{n})$, which yields the
improvement. However, the result of Borwein and Erd\'elyi~\cite{BorweinE97} used in previous work on trace reconstruction applies only to
univariate polynomials. Our key technical result is a generalization
of their result to accommodate bivariate Littlewood polynomials, which
we then use in a statistical test to identify the unknown matrix. 

 For~\pref{thm:random_matrix_intro}, using an averaging argument and exploiting randomness in
the original matrix, we construct a statistical test to determine if
two rows (or columns) from two different traces correspond to the same
row (column) in the original string. We show that this test succeeds
with overwhelming probability, which lets us align the rows and
columns in all traces. Once aligned, we know which rows/columns were
deleted from each trace, so we can simply read off the
original matrix $X$.

\paragraph*{Notation}
Throughout, $n$ is the length of the binary string being
reconstructed, $n_0$ is the number of \0s, $k$ is the number of \1s,
i.e., the \emph{sparsity} or \emph{weight}. For matrices, $n$ is the
total number of entries, and we focus on square
$\sqrt{n}\times\sqrt{n}$ matrices. For most of our results, we assume
 that $n, n_0, k$ are known since, if not, they can easily
be estimated using a polynomial number of traces. Let $p$ denote the
deletion probability when the \1s and \0s are deleted with the same
probability. We also study a channel where the \1s and \0s are deleted
with different probabilities; in this case, $p_0$ is the deletion
probability of a \0 and $p_1$ is the deletion probability of a \1. We
refer to the corresponding channel as the $(p_0,p_1)$-Deletion Channel
or the asymmetric deletion channel. It will also be
convenient to define $q=1-p, q_0=1-p_0$ and $q_1=1-p_1$ as the
corresponding retention probabilities. Throughout, $m$ denotes the
number of traces. {\color{black} For a natural number $w$ we use the notation $[w] = \{1,\ldots,w\}$.} 

\section{Sparsity and Learning Binomial Mixtures}
\label{sec:sparsity}

We begin with the sparse trace reconstruction problem, where we assume
that the unknown string $x$ has at most $k$ \1s.  Our analysis for
this setting is based on a simple reduction from trace reconstruction
to learning a mixture of binomial distributions, followed by a new
sample complexity guarantee for the latter problem. This approach
yields two new results: first, we obtain an
$\exp(O((k/q_1)^{1/3}\log^{2/3} n))$ sample complexity bound for sparse
trace reconstruction,
and second, we show that this guarantee applies even if the deletion probability for \0s is very close to $1$.

To establish our results, we introduce a slightly more challenging channel which we refer to as 
the \emph{Austere Deletion Channel}. 
The bulk of the proof analyzes
this channel, and we obtain results for the $(p_0,p_1)$ channel via a
simple reduction.

\begin{theorem}[Austere Deletion Channel Trace Reconstruction] \label{thm:austere}
In the Austere Deletion Channel, all but exactly one \0 are deleted
(the choice of which \0 to retain is made uniformly at random) and
each \1 is deleted with probability $p_1$.  For such a
channel, 
\[m=\exp({O}((k/q_1)^{1/3}\log^{2/3} n))\] 
traces suffice for
sparse trace reconstruction {\color{black}with high probability} where $q_1=1-p_1$, provided $q_1 =
  \Omega(\sqrt{k^{-1}\log n})$.
\end{theorem}

We will prove this result shortly, but we first derive our main result
for this section as a simple corollary. 

\begin{corollary}[Deletion Channel Trace Reconstruction]\label{corr:del}
For the $(p_0,p_1)$-deletion channel, \[m=q_0^{-1}
\exp({O}((k/q_1)^{1/3}\log^{2/3} n))\] traces suffice for sparse trace
reconstruction {\color{black}with high probability} where $q_0=1-p_0$ and $q_1=1-p_1=\Omega(\sqrt{k^{-1} \log n})$.
\end{corollary}
\begin{proof}
This follows from~\pref{thm:austere}. By focusing on just a single \0,
it is clear that the probability that a trace from the
$(p_0,p_1)$-deletion channel contains at least one \0 is at least
$q_0$. If among the retained \0s we keep one at random and remove the
rest, we generate a sample from the austere deletion channel. Thus,
with $m$ samples from the $(p_0,p_1)$ deletion channel, we obtain at
least $mq_0$ samples from the austere channel and the result follows.
Note that~\pref{thm:sparsity_intro} is a special case where
$p_0=p_1=p$. 
\end{proof}

\begin{remark} Note that the case where $q_1$ is constant (a typical setting for the problem) and $k=o(\log n)$ is not covered by the corollary. However, in this case a simpler approach applies to argue that $\poly(n)$ traces suffice: with probability $q_1^k\geq 1/\poly(n)$ no \1s are deleted in the generation of the trace and given $\poly(n)$ such traces, we can infer the original position of each \1 based on the average position of each $\1$ in each trace.
\end{remark}
\begin{remark}
Note that the weak dependence on $q_0$ ensures that as long as
$q_0=1/\exp({O}((k/q_1)^{1/3}\log^{2/3} n))$, we still have the 
$\exp({O}((k/q_1)^{1/3}\log^{2/3} n))$ bound. Thus, our result shows that sparse trace reconstruction is
possible even when zeros are retained with {\color{black}super-polynomially} small
probability.
\end{remark}

\subsection{Reduction to Learning Binomial Mixtures}
We prove~\pref{thm:austere} via a reduction {\color{black}from austere deletion channel trace reconstruction} to learning
binomial mixtures.  Given a string $x$ of length $n$, let $r_i$ be the
number of ones before the $i^{\textrm{th}}$ zero in $x$. For example, if
$x=1001100$ then $r_1=1, r_2=1,r_3=3, r_4=3.$ Note that the multi-set
$\{r_1, r_2, \ldots, \}$ uniquely determines $x$, that each $r_i\leq
k$, and that the multi-set has size $n_0$. The reduction from trace
reconstruction to learning binomial mixtures is
appealingly simple:

\begin{enumerate}
\item Given traces $t_1, \ldots, t_m $ from the austere
  channel, let $s_i$ be the number of leading ones in $t_i$.
\item Observe that each $s_i$ is generated by a uniform\footnote{Note
  that since the $r_i$ are not necessarily distinct some of the
  binomial distributions are the same.} mixture of $\bin(r_1,q_1),
  \ldots, \bin(r_{n_0},q_1)$ where $q_1=1-p_1$.  Hence, learning $r_1,
  r_2, \ldots, r_{n_0}$ from $s_1, s_2, \ldots, s_m$ allows us to
  reconstruct $x$. 
\end{enumerate}

{\color{black}We will say that a number $x$ has $t$-precision if $10^{y}\times x \in \mathbb{Z}$ where $y \in \mathbb{Z}$ and $y=O(\log t)$.}
To obtain~\pref{thm:austere}, we establish the following new
guarantee for learning binomial mixtures.

\begin{theorem}[Learning Binomial Mixtures]\label{thm:learningmix}
Let $\Mc$ be a mixture of $d=\poly(n)$ binomials:
\[
\mbox{Draw sample from $\bin(a_t,q)$ with probability $\alpha_t$}
\] where $0\leq a_1, \ldots, a_d \leq a$ are distinct integers, the values $\alpha_t$
have $\poly(n)$ precision, and $q=\Omega (\sqrt{a^{-1} \log n} )$.
Then $\exp({O}((a/q)^{1/3} \log^{2/3} n))$ samples suffice to learn the parameters exactly with high probability. 
\end{theorem}
\begin{proof}
Let $\Mc'$ be a mixture where the samples are drawn from $\sum_{t=1}^d \beta_t \bin(b_t,q)$, where $0 \leq b_1, \ldots, b_d\leq a$ are distinct and the probabilities $\beta_t \in \{0,\gamma,2\gamma, \ldots, 1\}$ where $1/\gamma=\poly(n)$. 
Consider the variational distance $\sum_t |A_t-B_t|$ between $\Mc$ and $\Mc'$ where
 \begin{align*}
A_t &= \prob{\mbox{sample from $\Mc$ is $t$}}  
= \sum_{j=1}^{d} \alpha_j {a_j \choose t} q^t (1-q)^{a_j-t} \\
B_t &= \prob{\mbox{sample from $\Mc'$ is $t$}} 
=\sum_{j=1}^{d} \beta_j {b_j \choose t}  q^t (1-q)^{b_j-t} \ .\end{align*} 
We will show that the variational distance between $\Mc$ and $\Mc'$ is at least 
\[\epsilon =\exp(-O((a/q)^{1/3} (\log 1/\gamma)^{2/3})) \ .\] Since there are at most $((a+1)\cdot( 1/\gamma+1))^d$ possible choices for the parameters of $\Mc'$, standard union bound arguments show that 
\begin{align*}
&O({\log (((a+1)\cdot( 1/\gamma+1))^d)}/{\epsilon^2}) = \exp(O((a/q)^{1/3} (\log 1/\gamma)^{2/3}))
\end{align*}
samples are sufficient to distinguish $\Mc$ from all other mixtures.

To prove the total variation bound, observe that by applying the binomial
formula, for any complex number $w$, we have
\begin{align*}
&\sum_{t\geq 0} (A_t-B_t)w^t 
= \sum_{t\geq 0} w^t \Big(\sum_{j\geq 0} \alpha_j  {a_j \choose t} q^i (1-q)^{a_j-t}  -  \beta_j{b_j \choose t} q^i (1-q)^{b_j-t} \Big) 
=  
\sum_{j\geq 0} (\alpha_j z^{a_j}-\beta_j z^{b_j})
\end{align*}
where $z=qw+(1-q)$.
Let $G(z)= \sum_{j\geq 0} (\alpha_j z^{a_j}-\beta_j z^{b_j})$ and apply the triangle inequality to obtain: 
\[
\sum_{t\geq 0}|A_t-B_t||w^t| \geq {|G(z)|} \ .
\]
Note that $G(z)$ is a non-zero degree $d$ polynomial with coefficients
in the set 
\[\{-1, \ldots, -2\gamma, -\gamma, 0 , \gamma, 2\gamma, \ldots, 1\} . \]  We would like to find a $z$ such that
$G(z)$ has large modulus but $|w^t|$ is small, since this will yield a
total variation lower bound. We proceed along similar lines to
Nazarov and Peres~\cite{NazarovP17} and De et al.~\cite{DeOS17}. It
follows from Corollary 3.2 in Borwein and Erd\'elyi~\cite{BorweinE97} that there exists $z\in
\{e^{i\theta} : -\pi/L\leq \theta \leq \pi/L\}$ such that 
\begin{align*}
|G(z)|\geq \gamma \exp(-c_1 L \log (1/\gamma))
\end{align*}
for some constant $c_1> 0$. For such a value of $z$, Nazarov and Peres~\cite{NazarovP17} show that
\begin{align*}
|w| \leq \exp(c_2/(qL)^2)
\end{align*}
for some constant $c_2>0$. Therefore,
\begin{align*}
&\sum_{t\geq 0}|A_t-B_t| \exp(t c_2/(qL)^2) \geq \sum_{t\geq 0}|A_t-B_t||w^t|  \geq {|G(z)|}  \geq \gamma \exp(-c_1L \log (1/\gamma))
\end{align*}

For $t > \tau=6qa$, by an application of the Chernoff bound, $A_t, B_t\leq 2^{-t}$, so we obtain
\begin{align*}
&\underbrace{\sum_{t>\tau}2^{-t}\exp(t c_2/(qL)^2)}_{= T_\tau}   + \sum_{t=0}^{\tau} |A_t-B_t| \exp(\tau c_2/(qL)^2) \geq \gamma \exp(-c_1 L \log (1/\gamma)) \ .
\end{align*}

\begin{align}\label{eq:bnd}
&\sum_{t=0}^\tau |A_t-B_t|  \geq  \frac{\gamma\exp(-c_1L \log (1/\gamma))}{\exp(\tau c_2/(qL)^2)}   - \frac{T_\tau}{\exp(\tau c_2/(qL)^2)}   \geq 
\frac{\gamma \exp(-c_1L \log (1/\gamma))}{\exp(\tau c_2/(qL)^2)} - O(2^{-\tau})
\end{align}
where the second equality follows from the assumption that $c_2/(qL^2)\leq (\ln 2)/2$ (which we will ensure when we set $L$) since,
\begin{align*}
&\frac{T_\tau}{\exp(\tau c_2/(qL)^2)}=\frac{O(1) \cdot 2^{-\tau}\exp(\tau c_2/(qL)^2)}{\exp(\tau c_2/(qL)^2)} = O(2^{-\tau}) \ .
\end{align*}
Set 
\[L=c \sqrt[3]{\tau/(q^2\log (1/\gamma))} = c \sqrt[3]{6a/(q \log (1/\gamma))}\] for some sufficiently large constant $c$. This ensures that
the first term of Eqn.~\ref{eq:bnd} is \[\exp(-O((a/q)^{1/3} \log^{2/3} (1/\gamma))) . \] Note that 
\begin{align*}
&\frac{c_2}{qL^2}
<\frac{c_2}{q c^2 (a/(q \log (1/\gamma)))^{2/3}} \leq \frac{c_2}{c^2 } \cdot  \left (\frac{\log (1/\gamma) }{aq^{1/2}} \right )^{2/3} 
\leq \frac{c_2}{c^2 } \cdot  \left (\frac{\log (1/\gamma) }{aq^2} \right )^{2/3} 
\end{align*}
and so by the assumption that $q=\Omega(\sqrt{\log (1/\gamma)/a})$ we may set the constant $c$ large enough such that $c_2/(qL^2)\leq (\ln 2)/2$ as required. The second term  of Eqn.~\ref{eq:bnd} is a lower order term given the assumption on $q$
and thus we obtain the required lower bound on the total variation distance.
\end{proof}

\pref{thm:austere} now follows from~\pref{thm:learningmix}, since in
the reduction, we have $d = O(n)$ binomials, one per \0 in $x$,
$\alpha_i$ is a multiple of $1/n_0$ and importantly, we have $a =
k$. The key is that we have a polynomial with degree $a=k$ rather than
a degree $n$ polynomial as in the previous analysis. 

\paragraph*{Remark} If all $\alpha_t$ are equal,~\pref{thm:learningmix} can be improved
to $\poly(n)\cdot \exp({O}((a/p)^{1/3}))$ by using a more refined bound from Borwein and Erd\'elyi~\cite{BorweinE97} in our proof. This follows by observing that if $\alpha_t=\beta_t=1/d$, then $\sum_{j\geq 0} (\alpha_j
z^{a_j}-\beta_j z^{s_j})$ is a multiple of a Littlewood polynomial and
we may use the stronger bound $|G(z)|\geq \exp(-c_1 L)/d$,
see Borwein and Erd\'elyi~\cite{BorweinE97}.%

\subsection{Lower Bound on Learning Binomial Mixtures} 
We now show that the exponential dependence on $a^{1/3}$
in~\pref{thm:learningmix} is necessary. 

\begin{theorem}[Binomial Mixtures Lower Bound]\label{thm:learningmix_lower}
There exists subsets \[\{a_1, \ldots, a_k\}\ne\{b_1, \ldots, b_k\}\subset \{0,\ldots ,a\}\] such that if $\Mc=\sum_{i=1}^{k}\bin(a_i,1/2)/k$ and $\Mc'=\sum_{i=1}^{k}\bin(b_i,1/2)/k$, then $\|\Mc-\Mc'\|_{TV} = \exp(-\Omega(a^{1/3} ))$. Thus, $\exp({\Omega}(a^{1/3}) )$ samples are required to distinguish $\Mc$ from $\Mc'$ with constant probability.
\end{theorem}

\begin{proof}
Previous work \cite{NazarovP17,DeOS17} shows the existence of two  strings $x,y\in \{0,1\}^n$ such that $\sum_i |t^x_i-t^y_i|= \exp(-\Omega(n^{1/3}))$
where $t^z_i$ is the expected value of the $i$th element ({\color{black}element at $i$th position counted from beginning}) of a string formed by applying the $(1/2,1/2)$-deletion channel to the string $z$. We may assume $\sum_{i\in [n]} x_i=\sum_{i\in [n]} y_i\equiv k$ since otherwise
\begin{align*}
&\sum_i \left |t^x_i-t^y_i \right |\geq \left |\sum_i  t^x_i- \sum_i  t^y_i \right | = \left |\sum_{i\in [n]} x_i/2-\sum_{i\in [n]} y_i/2 \right | \geq 1/2
\end{align*}
which would contradict the assumption $\sum_i |t^x_i-t^y_i| =  \exp(-\Omega(n^{1/3}))$. 

Consider $\Mc=\sum_{i=1}^{k}\bin(a_i,1/2)/k$ and $\Mc'=\sum_{i=1}^{k}\bin(b_i,1/2)/k$, where $a_i$ ($b_i$) is the number  of coordinates preceding the $i$th 1 in $x$ ($y$). 
Note that  
\[
t^x_i= \sum_{r=1}^k {a_r \choose i}/2^{a_r+1} ~~~\mbox{ and }~~~ t^y_i= \sum_{r=1}^k {b_r \choose i}/2^{b_r+1} \ ,
\]
and so
\begin{align*}
\|\Mc-\Mc'\|_{TV} =
\sum_i |\prob{\Mc=i}-\prob{\Mc'=i}|=
& \sum_i \frac{1}{k} \left | \sum_{r=1}^k {a_r \choose i}/2^{a_r}- 
  \sum_{r=1}^k {b_r \choose i}/2^{b_r} \right |  \\
=& \frac{2}{k} \sum_i |t^x_i-t^y_i| =\exp(-\Omega(n^{1/3})) \ ,
\end{align*}
which proves the result.
\end{proof}

\section{Well-Separated Sequences}
\label{sec:gaps}

We now prove~\pref{thm:gaps_intro}, showing that $\poly(n)$ traces suffice for
reconstruction of a $k$-sparse string when there are $\Omega(k\log n)$
\0s between each consecutive \1. For clarity of exposition, we are going to prove the statement of \pref{thm:gaps_intro} for $p=1/2$. The proof follows verbatim for any other constant $p$.
We call such sequences of \0s the
\emph{\0-runs} of the string. We also refer to the length of the
shortest \0-run as the \emph{gap} $g$ of the string $x$.

\begin{theorem*}[Restatement of~\pref{thm:gaps_intro}]
Let $x$ be a $k$-sparse string of length $n$ and gap at least $ck\log(n)$ for a large enough $c$. Then $\poly(n)$ traces from the
$(1/2,1/2)$-Deletion Channel suffice to recover $x$ with high
probability.
\end{theorem*} 

In~\pref{sec:overview}, we present a high-level overview of the algorithm and the analysis to provide intuition. 
In~\pref{sec:algorithm} we describe the algorithm in detail, state the key lemmas, and explain how to set the parameters. 
Due to the technical nature of the analysis, full details, including
proofs of the lemmas, are deferred to~\pref{app:letsgettechnical}.

\subsection{A Recursive Hierarchical Clustering Algorithm and Its Analysis: Overview} \label{sec:overview}

Let $\{p_u\}_{u=1}^k$ denote the positions (index of the coordinate from the left) of the $k$ \1s in the
original string $x$. Let $\mathcal{N}$ denote the multi-set of all
positions of all received \1s and call $N = |\mathcal{N}|$. We will
construct a graph $G$ on $N$ vertices where every vertex is associated
with a received \1. We decorate each vertex $v$ with a number $z_v
\in \mathcal{N}$, which is the position of the associated received
\1. Each vertex $v$ also has an \emph{unknown} label $y_v \in
\{1,\ldots,k\}$ denoting the corresponding \1 in the
\emph{original} string.

At a high level, our approach uses the observed values $\{z_v\}_{v \in
  V}$ to recover the unknown labels $\{y_v\}_{v \in V}$. 
  Once this ``alignment'' has been performed, the original
string can be recovered easily, since the average of $\{z_v
\mathbf{1}\{y_v = u\}\}_{v \in V}$ is an unbiased estimator for $p_u/2$.

\paragraph*{A starting observation} 
Our first observation
is a simple fact about binomial concentration, which we will use to
define the edge set in $G$: by the Chernoff bound, with high
probability, for every vertex $v$, if $y_v = u$ then we must have
$|z_v - p_u/2|\leq c\sqrt{n \log n}$ for some constant $c$. 
Defining the edges in $G$ to be $\{(v,w) : |z_v - z_w| \leq
2c\sqrt{n\log n}\}$ then guarantees that all vertices with $y_v=u$ are
connected. This immediately yields an algorithm for the much stronger
gap condition $g \geq 4c\sqrt{n \log n}$, since with such separation,
no two vertices $v,w$ with $y_v \ne y_w$ will have an edge. Therefore,
the connected components reveal the labeling so that $\poly(n)$ traces
suffice with $g = \Omega(\sqrt{n \log n})$.

Intuitively, we have constructed a clustering of the received \1s that
corresponds to the underlying labeling. To tolerate a weaker gap
condition, we proceed recursively, in effect constructing a
\emph{hierarchical clustering}. However there are many subtleties that
must be resolved.

\paragraph*{The first recursion}
To proceed, let us consider the weaker gap condition of $g \geq
\tilde{\Omega}(k^{1/2}n^{1/4})$. In this regime, $G$ still maintains a
\emph{consistency} property that for each $u$ all vertices with $y_v =
u$ are in the same connected component, but now a connected component
may have vertices with different labels, so that each connected
component $C$ identifies a continguous set $U \subset \{1,\ldots,k\}$
of the original \1s. Moreover, due to the sparsity assumption, $C$
must have length, defined as $\max_{v \in C}z_v - \min_{v \in C}z_v$,
at most $O(k\sqrt{n \log n})$. Therefore if we can correctly identify
every trace that contains the left-most and right-most \1 in $U$, we
can recurse and are left to solve a subproblem of
length $O(k\sqrt{n \log n})$. Appealing to our starting observation,
this can be done with
a gap of $g
\geq \tilde{\Omega}(k^{1/2}n^{1/4})$.

The challenge for this step is in identifying every trace that
contains the left-most and right-most \1 in $U$, which we call $u_L$
and $u_R$ respectively. This is important for ensuring a ``clean''
recursion, meaning that the traces used in the subproblem are
generated by passing exactly the same substring through the deletion
channel. To solve this problem we use a device that we call a
\emph{Length Filter}. For every trace, consider the subtrace that
starts with the first received \1 in $U$ and ends with the last
received \1 in $U$ (this subtrace can be identified using $G$). If the
trace contains $u_L,u_R$ then the length of this subtrace is
$2+\textrm{Bin}(L-2,1/2)$ where $L$ is the distance
between $u_L, u_R$ in the original string. On the other hand, if the
subtrace does not contain both end points, then the length is
$2+\textrm{Bin}(L'-2, 1/2)$ where $L' \leq L -
g$. Since we know that $L \leq \tilde{O}(k\sqrt{n})$ and we are
operating with gap condition $g = \tilde{\Omega}(k^{1/2}n^{1/4}) =
\tilde{\Omega}(\sqrt{L})$, binomial concentration implies that with
high probability we can \emph{exactly} identify the subtraces
containing $u_L$ and $u_R$. 

\paragraph*{Further recursion}
The difficulty in applying a second recursive step is that when $g =
o(k^{1/2}n^{1/4})$ the length filter cannot isolate the subtraces that
contain the leftmost and rightmost \1s for a block $U$, so we cannot
guarantee a clean recursion. However, substrings that pass through the
filter are only missing a short prefix/suffix which upper bounds any
error in the indices of the received \1s.  We ensure consistency at
subsequent levels by incorporating this error into a more cautious
definition of the edge set (in fact the additional error is the same
order as the binomial deviation at the next level, so it has
negligible effect). In this way, we can continue the recursion until
we have isolated each \1 from the original string. The
$\Omega(k\log n)$ lower bound on run length arises since the gap at
level $t$ of the recursion, $g_t$, is related to the gap at level
$t-1$ via $g_{t} = \sqrt{k \log n \cdot g_{t-1}}$ with $g_1 = \sqrt{n \log
  n}$, and this recursion asymptotes at
$\Omega(k\log n)$.

The last technical challenge is that, while we can isolate each
original \1, the error in our length filter introduces some bias into
the recursion, so simply averaging the $z_v$ values of the clustered
vertices does not accurately estimate the original position. However,
since we have isolated each \1 into pure clusters, for any connected
component corresponding to a block of \1s, we can identify \emph{all}
traces that contain the first and last \1 in the block. Applying this
idea recursively from the bottom up allows us to debias the recursion
and accurately estimate all positions.

\subsection{The algorithm in detail: recursive hierarchical clustering}
\label{sec:algorithm}
We now describe the recursive process in more detail. Let us define
the thresholds:
\begin{align*}
&\tau_1 = \tilde{O}(n^{1/2}), \tau_2 = \tilde{O}(k^{1/2}n^{1/4}),\tau_3 = \tilde{O}(k^{3/4}n^{1/8}),\ldots, \tau_D = \tilde{O}(k^{1-1/2^{(D-1)}}n^{1/2^D}),
\end{align*}
which will be used in the length filter and in the definitions of the
edge set. Observe that with $D= O(\log_2 \log_2 n)$, we have $\tau_D =
\tilde{O}(k)$.  Let $\tilde{x}_1,\ldots,\tilde{x}_m$ denote the $m =
\poly(n)$ traces.  We will construct a sequence of graphs
$G_1,G_2,\ldots,G_D$ on the vertex sets $V_1\supset V_2,\ldots,
\supset V_D$, where each vertex $v$ corresponds to a received \1 in
some trace $t_v \in [m]$ and is decorated with its position $z_v$ and the unknown label $y_v$.
The $d^{\textrm{th}}$ round of the
algorithm is specified as follows with $z_v^{(1)} = z_v$, $V_1$ as
the multi-set of all received \1s {\color{black} and $C_1^{(0)} = V_1$}.
\begin{enumerate}
\item Define $G_d$ with edge set
    $E_d = \bigcup_j \{ (v,w): v,w \in V_d {\color{black} \cap C_j^{(d-1)}  } \textrm{ and } |z_v^{(d)} - z_w^{(d)}| \leq \tau_d \}$.
\item Extract $k_d \leq k$ connected components $C^{(d)}_1,\ldots,C^{(d)}_{k_d}$
  from $G_d$. 
\item For each connected component $C^{(d)}_i$, extract subtraces
  $\{\tilde{x}^{(d,i)}_j\}_{j=1}^m$ where $\tilde{x}^{(d,i)}_j$ is the
  substring of $\tilde{x}_j$ starting with the first \1 in $C^{(d)}_i$ and
  ending with the last \1 in $C^{(d)}_i$. Formally, with $\ell = \min\{z_v : v \in
    C^{(d)}_i, t_v = j\}$ and $r = \max\{z_v : v \in C^{(d)}_i, t_v = j\}$, we
  define $\tilde{x}_j^{(d,i)} = \tilde{x}_j[\ell,\ldots,r]$.
\item Length Filter: Define $L^{(d,i)} = \max_{j}
  \textrm{len}(\tilde{x}_j^{(d,i)})$. If
  \begin{align*}
    \textrm{len}(\tilde{x}_j^{(d,i)}) \leq L^{(d,i)} - \Omega(\sqrt{L^{(d,i)}\log(L^{(d,i)})}),
  \end{align*}
  delete all vertices $v \in C^{(d)}_i$ with $t_v = j$. Let $V_{d+1}$ be the multi-set of all surviving vertices. 
\item For $v \in V_{d+1} \cap C^{(d)}_i$, define $z_v^{(d+1)} = z_v - \min_{v'
  \in C^{(d)}_i, t_v = t_{v'}}z_{v'}$. 
\end{enumerate}

\begin{algorithm}[t]
\caption{Algorithm \texttt{RecurGap}}
\label{alg:recursive}
\begin{algorithmic}
\REQUIRE Traces $\Tc = \{\tilde{x}_j\}_{j=1}^m$, gap lower bound $g \ge ck\log^2(n)$, levels of recursion $D$.
\STATE For each received \1, create vertex $v$ decorated with $(z_v, t_v)$ where $z_v \in [n]$ is the position of the received \1 and $t_v \in [m]$ is the index of the trace. 
\STATE {\color{black} Define thresholds $\tau_1 = 4\sqrt{2n\log(nmk)}$, $\tau_d = 80\sqrt{k\tau_{d-1}\log(mnk)}$ for $d=2,\ldots,D$.}
\STATE Set $z_v^{(1)} = z_v$, $V_1 = C_1^{(0)} = V$ 
\FOR{$d=1,\ldots,D$:}
\STATE Create graph $G_d$ with vertices $V_d$ and with edges
\begin{align*}
E_d = \bigcup_j \Big\{(v,w) \in V_d {\color{black} \cap C_j^{(d-1)} }: |z_v^{(d)} - z_w^{(d)}| \leq \tau_d/4\Big\}
\end{align*}
\STATE Extract connected components $C^{(d)}_1,\ldots,C^{(d)}_{k_d}$ of $G_d$. 
\STATE For each connected component $C^{(d)}_i$, extract subtraces
$\{\tilde{x}_j^{(d,i)}\}_{j=1}^m$ where $\tilde{x}_j^{(d,i)}
= \tilde{x}_j[\ell,r]$ and $\ell = \min\{z_v: v \in C_i^{(d)}, t_v =
j\}$ and $r = \max\{z_v :v \in C_i^{(d)}, t_v = j\}$.
\STATE Define $L^{(d,i)} = \max_j \textrm{len}(\tilde{x}_j^{(d,i)})$. If
\begin{align*}
        \textrm{len}(\tilde{x}_j^{(d,i)}) \leq L^{(d,i)} - {\color{black} 2}\sqrt{2L^{(d,i)}\log(kmn)}, 
\end{align*}
delete all vertices $v \in C^{(d)}_i$ with $t_v = j$. Let $V_{d+1}$ be all surviving vertices.
\STATE For $v \in C_i^{(d)} \cap V_{d+1}$, define $z_{v}^{(d+1)} = z_v - \min\{z_{v'} : v' \in C_i^{(d)}, t_{v'} = t_v\}$. 
\ENDFOR
\end{algorithmic}
\end{algorithm}

{\color{black}See~\pref{alg:recursive} for pseudocode. We note that
  $z_v^{(d)}$ corresponds to a shifted index of the received \1
  associated with vertex $v$. Intuitively, we shift by removing a
  prefix of the trace $t_v$, which provides a form of noise reduction.}

We analyze the procedure via a sequence of lemmas. The first one
establishes a basic consistency property: that two \1s originating
from the same source \1 are always clustered together.

\begin{lemma}[Consistency]
\label{lem:consistency}
At level $d$ let $V_{d,u} = \{v \in V_d, y_v = u\}$ for each $u \in
[k]$. Then with high probability, for each $d$ and $u$ there exists
some component $C^{(d)}_i$ at level $d$ such that $V_{d,u} \subset
C^{(d)}_i$.
\end{lemma}

The next lemma provides a length upper bound on any component, which
is important for the recursion. At a high level since we are using a
threshold $\tau_d$ at level $d$ and the string is $k$-sparse, no
connected component can span more than $k\tau_d$ positions.
\begin{lemma}[Length Bound]
\label{lem:length_bd}
At level $d$, the following holds with probability at least $1-1/n^2$:
For every component $C_i^{(d)}$ at level $d$, we have $L^{(d,i)} \leq
2k\tau_d$. Moreover if $U$ is a contiguous subsequence of
$\{1,\ldots,k\}$ with $\bigcup_{u \in U} V_{d,u} \subset C_i^{(d)}$,
then $| \min_{u \in U} p_u - \max_{u \in U} p_u | \leq 2k\tau_d$.
\end{lemma}

Finally we characterize the length filter.
\begin{lemma}[Length Filter]
\label{lem:filter}
Assume $m \geq n$. At level $d$, the following holds with probability
at least $1-1/n^2$: 
For a component $C_i^{(d)}$ at level $d$, let $U$ be the maximal
contiguous subsequence of $\{1,\ldots,k\}$ such that $\bigcup_{u \in
  U} V_{d,u} \subset C_i^{(d)}$. Define $u_L = \argmin_{u \in U} p_u$
and $u_R = \argmax_{u \in U} p_u$. Then for any $v \in C_i^{(d)}$, if
$u_L$ and $u_R$ are present in $t_v$, then $v$ survives to round
$d+1$, that is $v \in V_{d+1}$. Moreover, for any $v \in V_{d+1}$, let
$p_{\min}(v,U)$ denote the original position of the first \1 from $U$
that is also in the trace $t_v$. Then we have $p_{\min}(v,U) - p_{u_L}
\leq 8\sqrt{k\tau_d\log(nmk)}$.
\end{lemma}
The lemmas are all interconnected and proved formally
in~\pref{app:letsgettechnical}.  It is important that the error
incurred by the length filter is $\sqrt{k\tau_d} = \tau_{d+1}$ which
is exactly the binomial deviation at level $d+1$. Thus the threshold
used to construct $G_{d+1}$ accounts for both the length filter error
and the binomial deviation. This property, established
in~\pref{lem:filter}, is critical in the proof
of~\pref{lem:consistency}.

For the hierarchical clustering, observe that after $D = \log \log n$
iterations, we have $\tau_D = \tilde{O}(k)$. With gap condition $g =
\tilde{\Omega}(k)$ and applying~\pref{lem:consistency}, this means
that the connected components at level $D$ each correspond to exactly
one \1 in the original string. Moreover since the length filter
preserves every trace containing the left-most and right-most \1 in
the component, the probability that a subtrace passes through the
length filter is at least $1/4$. Hence, after $\log \log n$ levels,
the expected number of surviving traces in each cluster is $m/4^{\log
  \log n} = m/(\log^2n)$. Thus for each index $u \in \{1, \ldots, k\}$
corresponding to a $\1$ in the original string, our recursion
identifies at least $m/(\log^2n)$ vertices $v \in V_1$ such that $t_v
= u$.

\paragraph*{Removing Bias} 
The last step in the algorithm is to overcome {\color{black} the} bias introduced by the
length filter. The de-biasing process works upward from the bottom of
the recursion. Since we have isolated the vertices corresponding to
each \1 in the original string, for a component $C_i^{(D-1)}$ at level
$D-1$, we can identify all subtraces that survived to this level that
contain the first and last \1 of the corresponding block $U_i^{(D-1)}
\subset[k]$. Thus, we can eliminate all subtraces that erroneously
passed this length filter.

Working upwards, consider a component $C_i^{(d)}$ that corresponds to
a block $U_i^{(d)} \subset [k]$ of \1s in the original string. Since
we have performed further clustering, we have effectively partitioned
$U_i^{(d)}$ into sub-blocks $U_1^{(d+1)},\ldots,U_s^{(d+1)}$. We would
like to identify exactly the subtraces that survived to level $d$ that
contain the first and last \1 of $U_i^{(d)}$, but unfortunately this
is not possible due to a weak gap condition. However, by induction, we
\emph{can exactly} identify all subtraces that survive to level $d$
that contain the first and last \1 of the first and last sub-block of
$U_i^{(d)}$, namely $U_1^{(d+1)}$ and $U_s^{(d+1)}$. Thus we can
de-bias the length filter at level $d$ by filtering based on a more
stringent event, namely the presence of the $2^{D-d}$ nodes {\color{black} required to de-bias the first \emph{and} last blocks $U_1^{(d+1)}$ and $U_s^{(d+1)}$}. In total to
de-bias all length filters above a particular component, we require
the presence of $\sum_{d=1}^D 2^{D-d} = O(2^D) = O(\log n)$ nodes,
which happens with probability $\Omega(1/n)$. Thus we can debias with
only a polynomial overhead in sample complexity. 
See~\pref{fig:debias} for an illustration.

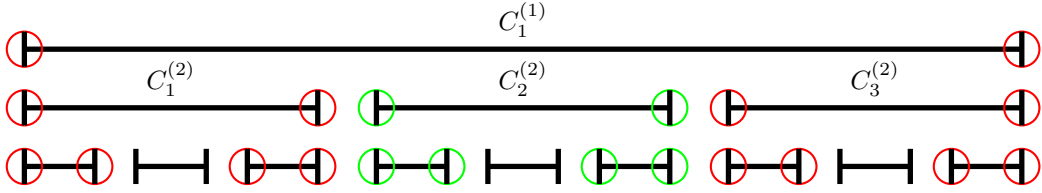
\begin{figure*}
\begin{center}
\begin{tikzpicture}[thick, scale=0.78]
 \filldraw[color=black, fill=red!0,  line width=2pt] (8.5,5) -- (25.5,5);
 \draw[color=red] (8.5, 5) circle (.3);
 \node[draw=none] at (17,5.5) {$C^{(1)}_{1}$};
 \node[draw=none] at (11,4.5) {$C^{(2)}_{1}$};
 \node[draw=none] at (17,4.5) {$C^{(2)}_{2}$};
 \node[draw=none] at (23,4.5) {$C^{(2)}_{3}$};
 \draw[color=red] (25.5, 5) circle (.3);
 \filldraw[color=black, fill=red!0,  line width=2pt] (8.5,4) -- (13.5,4);
 \draw[color=red] (8.5, 4) circle (.3);
 \draw[color=red] (13.5, 4) circle (.3);
 \filldraw[color=black, fill=red!0,  line width=2pt] 
(13.5,3.7) -- (13.5,4.3); 
 \filldraw[color=black, fill=red!0,  line width=2pt] (14.5,4) -- (19.5,4);
 \draw[color=green] (14.5, 4) circle (.3);
 \draw[color=green] (19.5, 4) circle (.3);
 \filldraw[color=black, fill=red!0,  line width=2pt] 
(14.5,3.7) -- (14.5,4.3); 
\filldraw[color=black, fill=red!0,  line width=2pt] 
(19.5,3.7) -- (19.5,4.3); 
 \filldraw[color=black, fill=red!0,  line width=2pt] (20.5,4) -- (25.5,4);
 \draw[color=red] (20.5, 4) circle (.3);
 \draw[color=red] (25.5, 4) circle (.3);
 \filldraw[color=black, fill=red!0,  line width=2pt] 
(20.5,3.7) -- (20.5,4.3); 
 \filldraw[color=black, fill=red!0,  line width=2pt] (8.5,3) -- (9.7,3);
 \draw[color=red] (8.5, 3) circle (.3);
 \draw[color=red] (9.7, 3) circle (.3);
 \filldraw[color=black, fill=red!0,  line width=2pt] 
(8.5,2.7) -- (8.5,3.3); 
 \filldraw[color=black, fill=red!0,  line width=2pt] 
(9.7,2.7) -- (9.7,3.3); 
\filldraw[color=black, fill=red!0,  line width=2pt] 
(10.4,3) -- (11.6,3);
 \filldraw[color=black, fill=red!0,  line width=2pt] 
(10.4,2.7) -- (10.4,3.3); 
  \filldraw[color=black, fill=red!0,  line width=2pt] 
(11.6,2.7) -- (11.6,3.3); 
\filldraw[color=black, fill=red!0,  line width=2pt] 
(12.3,3) -- (13.5,3);
\draw[color=red] (12.3, 3) circle (.3);
\draw[color=red] (13.5, 3) circle (.3);
 \filldraw[color=black, fill=red!0,  line width=2pt] 
(12.3,2.7) -- (12.3,3.3); 
  \filldraw[color=black, fill=red!0,  line width=2pt] 
(13.5,2.7) -- (13.5,3.3); 

 \filldraw[color=black, fill=red!0,  line width=2pt] (14.5,3) -- (15.7,3);
 \draw[color=green] (14.5, 3) circle (.3);
 \draw[color=green] (15.7, 3) circle (.3);
  \filldraw[color=black, fill=red!0,  line width=2pt] 
(14.5,2.7) -- (14.5,3.3); 
 \filldraw[color=black, fill=red!0,  line width=2pt] 
(15.7,2.7) -- (15.7,3.3); 
\filldraw[color=black, fill=red!0,  line width=2pt] 
(16.4,3) -- (17.6,3);
 \filldraw[color=black, fill=red!0,  line width=2pt] 
(16.4,2.7) -- (16.4,3.3); 
 \filldraw[color=black, fill=red!0,  line width=2pt] 
(17.6,2.7) -- (17.6,3.3); 
\filldraw[color=black, fill=red!0,  line width=2pt] 
(18.3,3) -- (19.5,3);
\draw[color=green] (18.3, 3) circle (.3);
\draw[color=green] (19.5, 3) circle (.3);
  \filldraw[color=black, fill=red!0,  line width=2pt] 
(18.3,2.7) -- (18.3,3.3); 
 \filldraw[color=black, fill=red!0,  line width=2pt] 
(19.5,2.7) -- (19.5,3.3); 
 \filldraw[color=black, fill=red!0,  line width=2pt] (20.5,3) -- (21.7,3);
 \draw[color=red] (20.5, 3) circle (.3);
 \draw[color=red] (21.7, 3) circle (.3);
  \filldraw[color=black, fill=red!0,  line width=2pt] 
(20.5,2.7) -- (20.5,3.3); 
 \filldraw[color=black, fill=red!0,  line width=2pt] 
(21.7,2.7) -- (21.7,3.3); 
\filldraw[color=black, fill=red!0,  line width=2pt] 
(22.4,3) -- (23.6,3);
 \filldraw[color=black, fill=red!0,  line width=2pt] 
(22.4,2.7) -- (22.4,3.3); 
  \filldraw[color=black, fill=red!0,  line width=2pt] 
(23.6,2.7) -- (23.6,3.3); 
\filldraw[color=black, fill=red!0,  line width=2pt] 
(24.3,3) -- (25.5,3);
\draw[color=red] (24.3, 3) circle (.3);
\draw[color=red] (25.5, 3) circle (.3);
 \filldraw[color=black, fill=red!0,  line width=2pt] 
(24.3,2.7) -- (24.3,3.3);
 \filldraw[color=black, fill=red!0,  line width=2pt] 
(25.5,2.7) -- (25.5,3.3); 
\filldraw[color=black, fill=red!0,  line width=2pt] 
(8.5,4.7) -- (8.5,5.3); 
\filldraw[color=black, fill=red!0,  line width=2pt] 
(25.5,4.7) -- (25.5,5.3); 
\filldraw[color=black, fill=red!0,  line width=2pt] 
(8.5,3.7) -- (8.5,4.3); 
\filldraw[color=black, fill=red!0,  line width=2pt] 
(25.5,3.7) -- (25.5,4.3); 
\end{tikzpicture}
\end{center}
\caption{
De-biasing of traces. The figure displays the regions of the original
string $x$ that correspond to each connected component found in the
algorithm. The end-points of each component correspond to \1s in the
original string. To de-bias the length filter for component
$C_1^{(1)}$ at level $1$, we identify and retain only the traces that
contain \emph{all} of the \1s colored red above. Then, to de-bias the
length filter at $C_2^{(2)}$ at level $2$, we identify and retain only
the traces that contain \emph{all} of the green \1s.}
\label{fig:debias}
\end{figure*}

\section{Applications of the Well-Separated Strings Result and Methodology}
\label{sec:wellsepapplications}

In this section, we present two applications of the results and methodology developed in the previous section. 

\subsection{Strengthening to a Parameterization by Runs} \label{sec:runs}
We next strengthen~\pref{thm:gaps_intro} to show that $\poly(n)$ traces suffice
under the assumption that each \0-run has length
$\tilde{\Omega}(r)$ where $r=1+|\{i\in [n-1]: x_i\neq x_{i+1}\}|$, in
the string $x$ being reconstructed. Observe that this is a weaker assumption than assuming $x$ has sparsity $k$ and each one is separated by a 0-run of length $\tilde{\Omega}(k)$ , since $r\leq 2k+1$ always, but $r$
can be much less than $k$.

\begin{theorem}
\label{thm:sparsity_with_runs}
For the $(1/2,1/2)$-Deletion Channel, $\poly(n)$ traces suffice {\color{black}with high probability} if the
lengths of the \0-runs are $\tilde{\Omega}(r)$ where $r$ is the number
of runs in $x$.
\end{theorem}

The proof is via a reduction to the
$k$-sparse case in the previous sections.
%
Let $x'\in \{0,1\}^{<n}$ be the string formed by replacing every run
of \1s in $x$ by a single \1. We first argue that we can reconstruct
$x'$ with high probability using $\poly(n)$ traces generated by
applying the $(1/2,1/2)$-Deletion Channel to $x$.
We will prove this result for the case $r=\Omega(\log n)$ since
otherwise $\poly(n)$ traces is sufficient even with no gap
promise.\footnote{Specifically, if $r=O(\log n)$, with probability at
  least $1/2^r=1/\poly(n)$ a trace also has $r$ runs. Given
  $\poly(n)$ traces with $r$ runs we can estimate each run length because we know the $i^\textrm{th}$ run in each such trace
  corresponds to the $i^{\textrm{th}}$ run in the original string.}
Observe that with $m=\poly(n)$ traces, if every \0-run in $x$ has
length at least $c\log n$ for some sufficiently large constant $c>0$,
then a bit in every \0-run of $x$ appears in every trace with high
probability. Conditioned on this event, no two \1's that originally
appeared in different runs of $x$ are adjacent in any trace. Next
replace each run of \1s in each trace with a single \1. The end result
is that we generate traces that are generated as if we had deleted
each \0 in $x'$ with probability $1/2$ and each \1 in $x'$ with
probability $1-1/2^t\geq 1/2$ where $t$ is the length of the run that
the \1 belonged to in $x$. This channel is not equivalent to the
$(1/2,1/2)$-Deletion channel, but our analysis for the sparse case ({\color{black}that only depends on the alignment of \1s using the deletion properties of the \0s})
continues to hold even if the deletion probability of each \1 is
different. Thus we can apply~\pref{thm:gaps_intro} to recover $x'$, and the
sparsity of $x'$ is at most $r$.
Since the algorithm identifies corresponding \1s in $x'$ in the
different traces, we can then estimate the length of the \1-runs in
$x$ that were collapsed to each single $\1$ of $x'$ by looking at the
lengths of the corresponding \1-runs in the traces of $x$ before they
were collapsed.


\subsection{Reconstruction of random sparse strings with polynomial traces}
Suppose we have an unknown string $x \in \{0,1\}^n$ such that every element of $x$ is sampled uniformly and independently according to $\textrm{Ber}(\eta)$ for some sufficiently small $\eta$. Again, we send $x$ through the deletion channel where every bit is deleted with probability $1/2$ and observe random traces. We have the following theorem characterizing the sufficient number of traces required to recover $x$.

\begin{theorem}
$\mathsf{poly}(n)$ traces are sufficient to recover $x \in \{0,1\}^n$ with high probability if every element of $x$ is drawn randomly according to $\textrm{Ber}(\eta)$ for $\eta \leq c/\sqrt{n\log n}$ where $c>0$ is some small constant.
\end{theorem}  

\begin{proof}
Let $\{p_u\}$ denote the positions (index of the coordinate from the left) of the  \1s in the
original string $x$. Let $\mathcal{N}$ denote the multi-set of all
positions of all received \1s and call $N = |\mathcal{N}|$. 
construct a graph $G$ on $N$ vertices where every vertex is associated
with a received \1. We decorate each vertex $v$ with a number $z_v
\in \mathcal{N}$, which is the position of the associated received
\1. Each vertex $v$ also has an \emph{unknown} label $y_v$ denoting the corresponding \1 in the
\emph{original} string. Finally, the edges in $G$ are defined as following: two vertices $v,w$ will have an edge if $\{(v,w) : |z_v - z_w| \leq
2a\sqrt{n\log n}\}$ for some appropriate large constant $a$. Consider the original string $x$ partitioned into $O(\sqrt{n})$ contiguous segments each of length $6a\sqrt{n\log n}$. In that case, notice that
\begin{align*}
&\Pr( \log n \text{ consecutive segments all include 1’s })  < \Big(1-\Big(1-\frac{c}{\sqrt{n\log n}}\Big)^{6a \sqrt{n \log n}}\Big)^{\log n} < (6ac)^{\log n}.
\end{align*}
Taking a union bound over all sets of $\log n$ consecutive segments ($O(\sqrt{n})$ of them), we get that no consecutive $\log n$ segments should all include \texttt{1}'s with probability at least $1-O(\sqrt{n}(6ac)^{\log n})$. We now have the following two claims:

\begin{claim}
For any two vertices $u,v$ such that $y_u \neq y_v$ and $|p_{y_u}-p_{y_v}| > 6a\sqrt{n \log n}$, they will never have an edge {\color{black}with high probability}.
\end{claim}
\begin{proof}
We will prove this claim by contradiction. Suppose $u,v$ indeed have an edge which must imply that $|z_u-z_v| \le 2a\sqrt{n\log n}\}$ because of the definition of graph $G$. Therefore we must have by using Chernoff bound 
\begin{align*}
\Pr(|z_u-\frac{p_{y_u}}{2}| \ge \frac{a\sqrt{n \log n}}{2} \le n^{-a^2/24}
\end{align*}
we can take a union bound over all vertices of the graph to conclude that $|z_u-\frac{p_{y_u}}{2}|\le 0.5a\sqrt{n \log n}$ for all vertices of the graph $G$. In that case,
\begin{align*}
&|\frac{p_u}{2}-\frac{p_{y_v}}{2}| \le |\frac{p_{y_v}}{2}- z_u|+|z_v-\frac{p_{y_v}}{2}|+|z_u-z_v|  \le 3a\sqrt{n \log n} 
\end{align*} 
which is a contradiction to the fact that $|p_{y_u}-p_{y_v}| > 6a\sqrt{n \log n}$.
\end{proof}
Therefore two \texttt{1}'s in the original string $x$ which are separated by at least $6a\sqrt{n \log n}$ will never have an edge in the graph $G$.

 \begin{claim}
 For $u,v \in \mathcal{N}$ such that $y_u=y_v$, there will exist an edge between $z_u$ and $z_v$ in the graph $G$ {\color{black}with high probability}.
 \end{claim}
 \begin{proof}
For two vertices $u,v \in \mathcal{N}$ such that $y_u=y_v$ (implying that $p_u=p_v$), we must have
\begin{align*}
|z_u-z_v| \le |z_u -\frac{p_{y_u}}{2}|+|z_v-\frac{p_{y_v}}{2}| \le 2a\sqrt{n \log n} 
\end{align*} 
with probability at least $1-n^{-a^2/6}$. Again, we can take a union bound over all vertices and over all traces to ensure that for $u,v \in \mathcal{N}$ such that $y_u=y_v$, there will exist an edge between $z_u$ and $z_v$ in the graph $G$. 
 \end{proof}

Further, the total number of \texttt{1}'s in a particular segment of the string $x$ of length $6a \sqrt{n \log^3 n}$, denoted by the random variable $X$ is sampled according to 
\begin{align*}
X \sim \mathsf{Bin}\Big(6a \sqrt{n \log^3 n}, \frac{c}{\sqrt{n \log n}}\Big).
\end{align*}
Therefore, we have $\avg X=6ac \log n$ and we can further use Chernoff Bound to conclude that $X \le 12ac \log n$ with probability at least $1-n^{-2ac}$. Taking a union bound, we can say that all segments of the string $x$ of length $6a \sqrt{n \log^3 n}$ has at most $12ac\log n$ \texttt{1}'s with probability of failure at most $n^{1-2ac}$. In that case, fix a particular connected component $C$ in the graph $G$ so that we can focus on reconstructing the contiguous sub-sequence of $x$ corresponding to the component $C$. From our previous analysis, we can ensure that 
\begin{align*}
\max_{u,v \in C} |p_{y_u}-p_{y_v}| \le 6a \sqrt{n \log^3 n}.
\end{align*}
since at most $\log n$ contiguous segments will include \texttt{1} in all of them. Moreover the total number of \texttt{1}'s in the component $C$ is at most $12ac\log n$. The probability that in a particular trace, all the \texttt{1}'s in the component $C$ will appear is at least $1/n^{12ac}$ and from now, we will only consider traces which has all the \texttt{1}'s present.  Subsequently, if the total number of traces used is $8n^{12ac+3}\log n$, then the number of traces containing all the \texttt{1}'s in $C$ is at least $8n^3$ with exponentially high probability. Using the Binomial Mean Estimator (defined in~\pref{app:letsgettechnical}), on these subset of traces containing all the \texttt{1}'s from $C$, we can recover the length of all the \texttt{0}-runs in the component $C$ with probability at least $1-n\exp(-n)$ ({\color{black}after taking union bound over at most $n$ \0 runs in $C$}). We can repeat this procedure to reconstruct the substrings of $x$ corresponding to all the components in the graph $G$. \\
In order to reconstruct the length of the run of \texttt{0}'s between two distinct components $C,C'$, we can only consider those traces where all the \texttt{1}'s corresponding to both $C,C'$ has appeared. There are at most $24ac\log n$ such \texttt{1}'s and as before, we can use $8n^{24ac+3}\log n$ traces to obtain $8n^3$ traces containing all the \texttt{1}'s in $C,C'$. Subsequently, using the Binomial Mean Estimator, we can reconstruct the length of the \texttt{0}-run between $C,C'$. Thus we can reconstruct the entire string with probability of failure at most $\sqrt{n}(6ac)^{\log n}+n^{1-2ac}+n^{24ac+3-a^2/24}+o(1/n)$. Setting $a,c$ appropriately results in a failure probability of $o(1)$.
\end{proof}

\section{Bounded Hamming Distance}
\label{app:hamming}
In this section, we turn to the sparse testing problem. We show that
it is possible to distinguish between two strings $x$ and $y$ with
Hamming distance $\Delta(x,y)<2k$, given $\exp(O(k\log n))$ traces.
This question is naturally related to sparse reconstruction, since the
difference string $x-y \in \{-1,0,1\}^n$ is at most $2k$ sparse, but
distinguishing two strings from traces is also at the core of our
analysis in~\pref{sec:sparsity}, as well as the analysis
of Nazarov and Peres~\cite{NazarovP17} and De et al.~\cite{DeOS17}. In particular given a
testing routine, reconstruction simply requires applying the union
bound.

In the binary symmetric channel (where each bit is flipped
independently with some probability), distinguishing between two
strings is easier if the Hamming distance is larger, since the two
strings are farther apart. However, it is unclear if this intuition
carries over to the deletion channel. In particular, the number of
traces required for testing is unlikely to even be monotonic in the
Hamming distance; if the Hamming distance is odd, then $x$ and $y$
have different Hamming weight, and we can estimate the Hamming weight
using just $O(n)$ traces.

{\color{black}
Our analysis uses a combinatorial result about \emph{$k$-decks} due
to Krasikov and Roditty~\cite{KRASIKOV1997344} that is defined below, along with an approach first used in McGregor et al.~\cite{McGregorPV14}.

\begin{definition}
The \emph{$k$-deck} of a string is the multi-set of all length $k$ subsequences of the string.
\end{definition}

\begin{theorem}[Krasikov and Roditty~\cite{KRASIKOV1997344}]\label{thm:kdecksuffice}
No two strings $x,y$ of length $n$ have the same $k$-deck if $\Delta(x,y)<2k$.
\end{theorem}
}
\begin{theorem}\label{thm:learningkdeck}
The $k$-deck of a binary string can be determined exactly with
$\exp(O(k\log n))$ traces from the symmetric deletion channel {\color{black}with high probability} assuming
$p\leq 1-k/n$.
\end{theorem}
\begin{proof}
We argue that sampling $\exp(O(k\log n))$ length $k$-subsequence of a
string is sufficient to reconstruct the $k$-deck with high
probability. The result then follows because if $p\leq 1-k/n$, then
with constant probability a trace generated by the deletion channel
has length at least $k$ and hence we can take a random $k$ subsequence
of such a trace as a random $k$ subsequence from $x$.
{\color{black}
Let $f_u$ be the number of times that $u\in \{\0,\1\}^k$ appears as a subsequence of $x$. Then, let $X_u$ be the number of times $u$ is generated if we sample $r=3n^{2k} \log n^k$ subsequences of length $k$ uniformly at random. $\expec{X_u}=rf_u / {n \choose k}$ and by an application of the Chernoff bound,
\begin{align*}
\prob{|X_u {n \choose k}/r-f_u|\geq 1} =\prob{|X_u-\expec{X_u}|\geq r/{n \choose k} } \leq \exp \left (- \frac{r}{3f_u{n \choose k}} \right )
\leq 1/n^k.
\end{align*}
where the last line follows given $ f_u\leq {n \choose k}$  and $r=3n^{2k} \log n^k$. Hence, by taking the union bound over all $2^k$ sequences $u$, it follows that we can determine the frequency of all length $k$ subsequences with high probability.
}
\end{proof}

\pref{thm:hamming_intro} follows directly from~\pref{thm:kdecksuffice}
and~\pref{thm:learningkdeck}.

\begin{theorem*}[Restatement of~\pref{thm:hamming_intro}]
For all $x, y\in \{\0,\1\}^n$ such that $\Delta(x,y)<2k$,  \[m=\exp(O(k\log n))\] traces are sufficient to  be  distinguished between $x$ and $y$ {\color{black}with high probability}.
\end{theorem*}

As noted earlier, if $\Delta(x,y)$ is odd then $\poly(n)$ traces suffice. Also, regardless of the Hamming distance, if the location of the first and second positions (say $i$ and $j$) where $x$ and $y$ differs by at least $\Omega(\sqrt{n\log n})$ then it is easy to show that expected weight of the length $i/2$ prefix of the traces differs by $\Omega(1/\poly(n))$ and hence we can distinguish $x$ and $y$ with $\poly(n)$ traces.

\section{Reconstructing Arbitrary Matrices}{\label{sec:matrices}}
Recall that in the matrix reconstruction problem, we are given samples
of a matrix $X \in \{0,1\}^{\sqrt{n}\times\sqrt{n}}$ passed through a
\emph{matrix deletion channel}, which deletes each row and each column
independently with probability $p=1-q$.  In this section we
prove~\pref{thm:matrix_intro}.

\begin{theorem*}[Restatement of~\pref{thm:matrix_intro}]
For matrix reconstruction, $\exp(O(n^{1/4}\sqrt{p\log n}/q))$ traces
suffice {\color{black}with high probability} to recover an arbitrary matrix $X \in
\{0,1\}^{\sqrt{n}\times\sqrt{n}}$, where $p$ is the deletion
probability and $q=1-p$.
\end{theorem*}

 The bulk of the proof involves designing a procedure to
test between two matrices $X$ and $Y$. This test is based on
identifying a particular received entry where the traces must differ
significantly, and to show this, we analyze a certain bivariate
Littlewood polynomial, which is the bulk of the proof. Equipped with
this test, we can apply a union bound and simply search over all pairs
of matrices to recover the string.

For a matrix $X \in \{0,1\}^{\sqrt{n}\times\sqrt{n}}$, let $\tilde{X}$
denote a matrix trace.  Let us denote the $(i,j)^{\textrm{th}}$ entry
of the matrix as $X_{i,j}, i,j=0,1, \dots, \sqrt{n}-1$, an indexing
protocol we adhere to for every matrix.  For two complex numbers
$w_1,w_2 \in \mathbb{C}$, observe that
\begin{align*}
\mathbb{E}\left[\sum_{i,j=0}^{\sqrt{n}-1}\tilde{X}_{i,j}w_1^iw_2^j\right]  
&= q^2\sum_{i,j}w_1^iw_2^j\sum_{k_i\geq i, k_j \geq j} X_{k_i,k_j}{k_i \choose i} \times  {k_j \choose j} p^{k_i-i}q^{i} p^{k_j-j}q^{j}\\
& = q^2\sum_{k_1,k_2=0}^{\sqrt{n}-1} X_{k_1,k_2}(qw_1 + p)^{k_1}(qw_2 + p)^{k_2} \ .
\end{align*}
Thus, for two matrices $X,Y$, we have
\begin{align*}
\frac{1}{q^2} \mathbb{E}\left[\sum_{i,j=0}^{\sqrt{n}-1}(\tilde{X}_{i,j}-\tilde{Y}_{i,j})w_1^iw_2^j\right]  = \sum_{k_1,k_2=0}^{\sqrt{n}-1} (X_{k_1,k_2}-Y_{k_1,k_2})(qw_1 + p)^{k_1}(qw_2 + p)^{k_2} \triangleq A(z_1,z_2)
\end{align*}
where we are rebinding $z_1 = qw_1+p$ and $z_2 = qw_2+p$. Observe that
$A(z_1,z_2)$ is a \emph{bivariate Littlewood polynomial}; all
coefficients are in $\{-1,0,1\}$, and the degree is {\color{black}$\sqrt{n}-1$} in each variable. For
such polynomials, we have the following estimate, which extends 
a result due of Borwein and Erd\'elyi~\cite{BorweinE97} for univariate
polynomials. 

\begin{lemma}
\label{lem:bivariate_littlewood}
Let $f(z_1,z_2)$ be non-zero Littlewood polynomial of {\color{black}degree $\sqrt{n}-1$ in each variable}. Then,
\begin{align*}
|f(z_1^\star,z_2^\star)| \geq \exp(-C_1L^2 \log n)
\end{align*}
for some $z_1^\star=\exp( i\theta_1), z_2^\star=\exp( i\theta_2)$
where $|\theta_1|, |\theta_2|\leq \pi/L$, and $C_1$ is a
universal constant.
\end{lemma}

\begin{proof}
Fix $L > 0$ and define the polynomial
\begin{align*}
F(z_1,z_2)=\prod_{1\leq a\leq L,1\leq b\leq L} f( z_1 e^{\pi i a /L}, z_2 e^{\pi i b /L}).
\end{align*}
{\color{black}
We use the maximum modulus principle that is stated as follows: For any holomorphic function $f$, the modulus of $f$ i.e. $|f|$ does not have a strict local maxima completely within its domain and therefore achieves the maximum value on the boundary of its domain. We first show by an iterated application of the maximum modulus principle that 
there exists $z^\star_1,z^\star_2$ on the unit disk such that
$F(z^\star_1,z^\star_2) \geq 1$}. First factorize $F(z_1,z_2) = z_2^k
G(z_1,z_2)$ where $k$ is chosen such that $G(z_1,z_2)$ has no common
factors of $z_2$. Since $F$ has non-zero coefficients, this implies
that $G(z_1,0)$ is a non-zero univariate polynomial. Further factorize
$G(z_1,0) = z_1^{\ell}H(z_1)$ so that terms in $H$ have no common
factors of $z_1$. $H$ is also a Littlewood polynomial and moreover it
has non-zero leading term, so that $|H(0)| \geq 1$. Thus by the
maximum modulus principle:
\begin{align*}
&|F(z^\star_1,z^\star_2)| = |G(z^\star_1,z^\star_2)| \geq |G(z^\star_1,0)| = |H(z^\star_1)| \geq |H(0)| \geq 1.
\end{align*}
Now, for any $a,b \in \{1,\ldots,L\}$ we have
{\color{black}
\begin{align*}
1 \leq |F(z_1^\star,z_2^\star)| \leq |f(z_1^\star e^{\pi ia/L}, z_2^\star e^{\pi i b/L})| \cdot n^{(L^2-1)},
\end{align*}
where we are using the fact that $|f(z_1,z_2)| \leq n$. This proves the lemma, since we may choose $a,b$ such that $z_1^\star e^{\pi ia /L} = \exp(i\theta_1),z_2^\star e^{\pi ib /L} = \exp(i\theta_2)$ for $|\theta_1|,|\theta_2| \leq \pi/L$. 
}
\end{proof}

Let $\gamma_L = \{e^{i\theta}: |\theta| \leq \pi/L\}$ denote the arc
specified in~\pref{lem:bivariate_littlewood}. For any $z_1 \in
\gamma_L$, Nazarov and Peres~\cite{NazarovP17} provide the following estimate for the modulus of
$w_1 = (z_1-p)/q$:
\begin{align*}
\forall z \in \gamma_L: |(z-p)/q| \leq \exp(C_2p/(Lq)^2).
\end{align*}
Using these two estimates, we may sandwich $|A(z_1,z_2)|$ by
\begin{align*}
\exp(-C_1L^2\log n) \leq \max_{z_1,z_2 \in \gamma_L} |A(z_1,z_2)|   \leq \frac{\exp(C'p\sqrt{n}/(Lq)^2)}{q^2}\sum_{ij} \left|\EE[\tilde{X}_{ij} - \tilde{Y}_{ij}]\right|  \ . 
\end{align*}
This implies that there exists some coordinate $(i,j)$ such that
\begin{align*}
\left|\EE[\tilde{X}_{ij} - \tilde{Y}_{ij}]\right| \geq \frac{q^2}{n} \exp\left(-C_1L^2\log n - \frac{C'p\sqrt{n}}{L^2q^2}\right)  \geq \frac{q^2}{n} \exp\left(-C \frac{n^{1/4}\sqrt{p\log n}}{q}\right)
,
\end{align*}
where the second inequality follows by optimizing for $L$.

The remainder of the proof follows the argument of~\citep{NazarovP17}:
Since we have witnessed significant separation between the traces
received from $X$ and those received from $Y$, we can test between
these cases with $\exp(O(n^{1/4}\sqrt{\log n}))$ samples (via a simple
Chernoff bound). Since we do not know which of the $2^{n}$ {\color{black}matrices} is
the truth, we actually test between all pairs, where the test has no
guarantee if neither matrix is the truth. However, via a union bound,
the true matrix will beat every other in these tests and this only
introduces a $\poly(n)$ factor in the sample complexity. 

\section{Reconstructing Random Matrices}{\label{sec:rand}}

In this section, we prove~\pref{thm:random_matrix_intro}: $O(\log n)$ traces suffice to reconstruct a random $\sqrt{n}\times
\sqrt{n}$ matrix with high probability for any constant deletion
probability $p<1$. This is optimal since $\Omega(\log n)$ traces are
necessary to just ensure that with high probability, every bit appears in at least one trace.

Our result is proved in two steps. We first design an oracle that allows us to identify when two rows (or two columns) in different matrix traces correspond to the same row (resp.~column) of the original matrix. We then use this oracle to identify which rows and columns of the original matrix have been deleted to generate each trace. This allows us to identify the original position of each bit in each trace. Hence, as long as each bit is preserved in at least one trace (and $O(\log n)$ traces is sufficient to ensure this with high probability), we can reconstruct the entire original matrix. 

\subsection{Steps to reconstruct the matrix}

\paragraph{Oracle for Identifying Corresponding Rows/Columns} We will first design an oracle that given two strings $t$ and $t'$ distinguishes, for any constant $q>0$, with high probability between the cases:

\begin{description}
\item[Case 1:]~ $t$ and $t'$ are traces generated by the deletion channel with preservation probability $q$ from the same random string $x\in_R \{0,1\}^{\sqrt{n}}$
\item[Case 2:]~ $t$ and $t'$ are traces generated by the deletion channel with preservation probability $q$ from independent random strings $x,y\in_R \{0,1\}^{\sqrt{n}}$
\end{description}

It $t$ and $t'$ are two rows (or two columns) from two different matrix traces, then this test determines whether $t$ and $t'$ correspond to the same or different row (resp.~column) of the original matrix.
In~\pref{sec:oracle}, we show how to perform this test with failure probability at most $1/n^{10}$. 
In fact, the failure probability can be made exponentially small but a polynomially small failure probability will be sufficient for our purposes.

\paragraph{Using the Oracle for Reconstruction}
Given $m=\Theta(\log n)$ traces we can ensure that every bit of $X$ appears in at least one of the matrix traces with high probability. We then use this oracle to associate each row in each trace with the rows in other traces that are subsequences of the same original row. This requires at most $\binom{m\sqrt{n}}{2}\leq (m\sqrt{n})^2$ applications of the oracle and so, by the union bound, this can performed with failure probability  at most $(m\sqrt{n})^2/n^{10}\leq 1/n^8$ where the inequality applies for sufficiently large $n$.

After using the oracle to identify corresponding rows amongst the different traces we group all the rows of the traces into $\sqrt{n}$ groups $G_1, \ldots, G_{\sqrt{n}}$ where the expected size of each group is $mq$. 
We next infer which group corresponds to the $i^{\textrm{th}}$ row of $X$ for each $i \in [\sqrt{n}]$.
Let $f$ be the bijection between groups and $[\sqrt{n}]$ that we are trying to learn, i.e., $f(j)=i$ if the $j^{\textrm{th}}$ group corresponds to the $i^{\textrm{th}}$ row of $X$.  If suffices to determine whether $f(j)<f(j')$ or $f(j)>f(j')$ for each pair $j\neq j'$. If there exists a matrix trace $\tilde{X}$ that includes a row in $G_j$ and a row in $G_{j'}$ then we can infer the relative ordering of  $f(j)$  and $f(j')$ based on whether the row from $G_j$ appears higher or lower in $\tilde{X}$ than the row in $G_{j'}$. The probability there exists such a trace is $1-(1-q^2)^m\geq 1-1/\poly(n)$ and we can learn the bijection $f$ with high probability.

We also perform an analogous process with columns. After both rows and columns have been processed, we know exactly which rows and columns were deleted to form each trace, which reveals the original position of each received bit in each trace. Given that every bit of $X$ appeared in at least some trace, this suffices to reconstruct $X$, proving~\pref{thm:random_matrix_intro}.

\begin{theorem*}[Restatement of~\pref{thm:random_matrix_intro}]
For any constant deletion probability $p<1$, $O(\log n)$ traces are sufficient to reconstruct a random $X\in \{0,1\}^{\sqrt{n}\times \sqrt{n}}$ {\color{black}with high probability}.
\end{theorem*}

\subsection{Oracle: Testing whether two traces come from same random string}
\label{sec:oracle}

For any $i \in \{0,1,\dots,\lfloor n/2w \rfloor \}$, define $S_i = \{2wi+j:j=0,\ldots, w-1\}$ to be a contiguous subset of size \[w=100n^{1/4} \sqrt{1/q\cdot \log n} \ .\]  Note that  there are size $w$ gaps between each $S_i$ and $S_{i+1}$, i.e., $w$ elements that are both larger than $S_i$ and smaller than $S_{i+1}$. This will later help us argue that  the bits in positions $S_i$ and $S_{i+1}$ in different traces are independent.
Given traces $t, t'$, define the three quantities:
$
X_i=\sum_{j\in S_i} t_{j}$, $ Y_i=\sum_{j\in S_i} t'_{j}$ and $
Z_i=(X_i-Y_i)^2$.
We will show that by considering $Z_0, Z_1, Z_2, \ldots$ we can determine whether $t$ and $t'$ are traces of the same original string or traces of two different random strings.

The basic idea is that if $t$ and $t'$ are generated by the same string, many of the bits summed to construct $X_i$ and the bits summed to construct $Y_i$ will correspond to the same bits of the original string; hence $Z_i$ will be smaller than it would be if $t$ and $t'$ were generated from two independent random strings.  
To make this precise, we need to introduce some additional notation.

\begin{definition}
For $A\subset \{0,1,2, \ldots \}$, let $R_t(A)$ be the indices of the bits in the transmitted string that landed in positions $A$ in trace $t$. Similarly define $R_{t'}(A)$.  
For example, if  bits in position 0 and 2 were deleted during the transmission of $t$ then 
$R_t(\{0,1,2\})=\{1,3,4\}$.
\end{definition}

The next lemma quantifies the overlap between $R_t(S_i)$ and $R_{t'}(S_i)$.

\begin{lemma}[Deletion Patterns]\label{lem:delpat}
With high probability over the randomness of the deletion channel, 
\begin{align*}
\forall i~,~ |R_t(S_i)\cap R_{t'}(S_i)| \geq q w/2
\quad \mbox{ and } \quad \quad \forall i\neq j~,~ |R_t(S_i)\cap R_{t'}(S_j)|= 0 \ .
\end{align*}
Note that  conditioned on the second property, each of the $Z_i$'s are independent random variables.
\end{lemma}
\begin{proof}
First note that by the Chernoff bound, for each $j\in [\sqrt{n}]$, the $j^{\textrm{th}}$ bit of the original sequence appears in position that belongs to $[qj- r,qj-r+1,\dots,qj+r-1,qj+r]$ where $r=5n^{1/4} \sqrt{q\log n}$ with high probability. The second part of the lemma follows since $r= wq /20 < w/20$ and therefore, with high probability, any bit in the original string will not appear in $S_\alpha$ in one trace and $S_\beta$ in another for $\alpha \neq \beta$ because there was a size $w$ gap between $S_\alpha$ and $S_\beta$ .

For the first part of the lemma, for each $S_i$, define 
\begin{align*}
S'_i=\Big\{\frac{2wi}{q}+\frac{r}{q}, \frac{2wi}{q}+\frac{r+1}{q}, \ldots, \frac{2wi+w-1}{q}-\frac{r}{q} \Big\}  \ .
\end{align*}
By the Chernoff Bound, with high probability the $w/q-2r/q> 0.9w/q$ bits in $S'_i$ positions in the original string arrive in positions $S_i$ in the trace. Also with high probability, $0.9 q^2 |S'_i|$ of the bits in $S'_i$ are transmitted in the generation of both $t$ and $t'$. Hence,
$
|R_t(S_i)\cap R_{t'}(S_i)| \geq  0.9w/q \cdot 0.9 q^2 > qw/2
$
as required. 
\end{proof}

{\color{black}Now, we prove a helper lemma characterizing the mean and variance of the square of difference of two independent binomials.} 

\begin{lemma}\label{lem:binomialdiff}
Let $A \sim \bin(h,1/2)$ and $B\sim \bin(h,1/2)$ be independent and $C=(A-B)^2$. Then, 
\[ \EE[C]=h/2 \quad \mbox{ and } \quad \Var[C]\leq h^2/2 \ .\]
\end{lemma}
\begin{proof}
The result follows by direct calculation:
\begin{align*}
&\EE[(A-B)^2]=\EE[A^2]+\EE[B^2]-2\EE[A]\EE[B] =h(h+1)/2-h^2/2=h/2
\end{align*}
and 
\begin{align*}
&\Var((A-B)^2)=\EE[(A-B)^4]-(h/2)^2 =h/2+6 {h \choose 2}/4-h^2/4= (2h-1)h/4 \ .\tag*\qedhere
\end{align*}
\end{proof}

We are now ready to argue that the values $Z_0, Z_1, \ldots $ are sufficient to determine whether or not $t$ and $t'$ are generated from the same random string. 

\begin{theorem}
\label{thm:median_estimator}
Let $z_j=\sum_{i=0}^{g-1} Z_{jg+i}$ for $g=96/q^2$ and $D=\textup{median}(z_0, z_1, z_2, \ldots , z_{\Theta (\log n)} )$.
\begin{description}
\item[Case 1.] \hspace{0.1in} If $t$ and $t'$ are generated from the same string, then $\Pr[D< (1-q/4)gw/2]\geq 1-1/n^{10}$.
\item[Case 2.] \hspace{0.1in} If $t$ and $t'$ are generated from different strings, then $\Pr[D\geq (1-q/4)gw/2] \geq 1-1/n^{10}$. 
\end{description}
\end{theorem}

\begin{proof}
Throughout the proof we condition on the equations in~\pref{lem:delpat} being satisfied. Note that this event is a function of the randomness of the deletion channel rather than the randomness of the strings being transmitted over the deletion channel.

First, suppose $t$ and $t'$ are generated from different strings. Then $Z_i$ has the same distribution as the variable $C$ in~\pref{lem:binomialdiff} when $r$ is set to $w$. Hence, $\EE[z_j]=gw/2$ and $\Var(z_j)\leq gw^2/2$. Therefore,
\begin{align*}
&\Pr[z_j< (1-q/4)gw/2]\leq \Pr[|z_j-\EE[z_j]|\geq (q/4)gw/2]  \leq \frac{\Var(z_j)}{\EE[z_j]^2 \cdot q^2/16} \leq \frac{2}{g q^2/16}=1/3.
 \end{align*}
 Therefore, by the Chernoff bound,  $D\geq (1-q/4)gw/2$ with  probability at least $1-1/n^{10}$.
 
Now, suppose $t$ and $t'$ are generated from the same string. Then, $Z_i$ has the same distribution as $C$ in~\pref{lem:binomialdiff} for  some $r\leq w-qw/2$. Hence, $\EE[z_j]=gr/2$ and $\Var(z_j)\leq gr^2/2$. Therefore,
\begin{align*}
&\Pr[z_j\geq (1-q/4)gw/2] 
\leq \Pr[|z_j-\EE[z_j]|\geq (q/4)gw/2] \leq \frac{\Var(z_j)}{\EE[z_j]^2 \cdot q^2/16} \leq \frac{2}{g q^2/16}=1/3.
 \end{align*}
 Therefore, by the Chernoff bound,  $D<  (1-q/4)gw/2$ with  probability at least $1-1/n^{10}$.
\end{proof}

\section{Extending Matrix Results to Tensors}
\label{sec:tensors}
\subsection{Reconstruction of arbitrary tensors}

In this setting, we have a $k^{th}$ order binary tensor
$T \in \{0,1\}^{n^{1/k} \times n^{1/k}\times \dots \times n^{1/k}}$ such that $T$ has equal number of elements along every dimension. The tensor $T$ is now passed through a
\emph{tensor deletion channel}, which deletes each element along every dimension 
independently with probability $p=1-q$. Notice that this is a generalization of the previous settings in matrix reconstruction (special case for $k=2$) and the trace reconstruction problem (special case for $k=1$) considered earlier.

In this section we
prove~\pref{thm:tensor_intro}.

\begin{theorem}{\label{thm:tensor_intro}}
For tensor reconstruction, $\exp\Big(O\Big( (n(kp/q^2)^k\log^2 n)^{1/(k+2)} \Big)\Big)$ traces
suffice {\color{black}with high probability} to recover an arbitrary tensor $T \in
\{0,1\}^{n^{1/k} \times n^{1/k}\times \dots \times n^{1/k}}$, where $p$ is the deletion
probability and $q=1-p$.
\end{theorem}

We again design a procedure to
test between two tensors $T_1$ and $T_2$. This test is based on
identifying a particular received entry where the traces (traces of the two tensors) must differ
significantly, and to show this, we analyze a certain multivariate
Littlewood polynomial. Equipped with
this test, we can apply a union bound and simply search over all pairs
of tensors to recover the correct one. We will begin by showing an extension of Lemma \ref{lem:bivariate_littlewood} for any value of $k$. 
\begin{lemma}{\label{lem:multivariate_littlewood}}
Let $f(z_1,z_2,\dots,z_k)$ be a non-zero Littlewood polynomial of {\color{black}degree $n^{1/k}$ in each variable}. In that case, 
\begin{align*}
\left|f(z_1^{\star},z_2^{\star},\dots,z_k^{\star})\right| \ge \exp(-C_1L^k \log n)
\end{align*}
for some $z_1^{\star}=\exp(i\theta_1),z_2^{\star}=\exp(i\theta_2),\dots,z_k^{\star}=\exp(i\theta_k)$ where $|\theta_1|,|\theta_2|,\dots,|\theta_k| \le \pi/L$ and $C_1$ is a universal constant.
\end{lemma}

The proof of Lemma \ref{lem:multivariate_littlewood} follows from an iterative use of the maximum modulus principle for multivariate Littlewood polynomials and follows along the lines of the proof presented in Lemma \ref{lem:bivariate_littlewood}. The detailed proof has been deferred to Appendix \ref{app:misstensors}.

For a matrix $T \in
\{0,1\}^{n^{1/k} \times n^{1/k}\times \dots \times n^{1/k}}$, let $\tilde{T}$
denote a tensor trace (the output after the tensor $T$ is passed through the \emph{tensor deletion channel}).  
Let us denote by $T_{i_1,i_2,\dots,i_k}$ the element in $T$ whose location along the $j^{\textrm{th}}$ dimension is $i_j+1$ i.e. there are $i_j$ elements along the $j^{\textrm{th}}$ dimension before $T_{i_1,i_2,\dots,i_k}$. Notice that this indexing protocol uniquely determines the element within the tensor. We now show the following lemma:

\begin{lemma}{\label{lem:distincttensors}}
For any two distinct tensors $T_1,T_2$, there exists a position denoted by the set of ordered indices $i_1,i_2,\dots,i_k$ such that 
\begin{align*}
\left|\EE[\tilde{T_1}_{i_1,i_2,\dots,i_k} - \tilde{T_2}_{i_1,i_2,\dots,i_k}]\right| \geq \frac{q^k}{n} \exp\left(-C n^{1/(k+2)} \Big( \frac{kp \log^{2/k} n}{q^2} \Big)^{k/(k+2)}\right),
\end{align*}
\end{lemma}

The proof of Lemma \ref{lem:distincttensors} follows from using the complex generating function of the tensor traces and subsequently, using Lemma \ref{lem:multivariate_littlewood} based on similar ideas as in Section \ref{sec:matrices}. The detailed proof has been deferred to Appendix \ref{app:misstensors}. For the remaining part, we follow the argument of~\citep{NazarovP17}:
Since we have witnessed significant separation between the traces
received from $X$ and those received from $Y$, we can test between
these cases with $\exp(O( (nk^k\log^2 n)^{1/(k+2)} ))$ samples (via a simple
Chernoff bound). Since we do not know which of the $2^{n}$ traces is
the truth, we actually test between all pairs, where the test has no
guarantee if neither tensor is the truth. However, via a union bound,
the true tensor will beat every other in these tests and this only
introduces a $\poly(n)$ factor in the sample complexity. 

\subsection{Reconstruction of random tensors}
In this section, we extend the results in Section \ref{sec:rand} for random tensors. Suppose we have a $k^{th}$ order random binary tensor
$T \in \{0,1\}^{n^{1/k} \times n^{1/k}\times \dots \times n^{1/k}}$ such that $T$ has equal number of elements along every dimension and every element in $T$ is randomly sampled from $\{0,1\}$ uniformly and independently. The tensor $T$ is now passed through a
\emph{tensor deletion channel}, which deletes each element along every dimension 
independently with probability $p=1-q$. In this section we will prove the following theorem:

\begin{theorem}\label{thm:randomtensors}
For any constant deletion probability $p<1$, $O(\log n/(1-p)^k)$ traces are sufficient {\color{black}with high probability} to reconstruct a random $X\in \{0,1\}^{n^{1/k} \times n^{1/k}}$.
\end{theorem}
Notice that this bound is also tight since we need $\Omega(\log n/(1-p)^k)$ traces to at least observe every bit in the tensor $T$. The detailed proof of Theorem \ref{thm:randomtensors} is a generalization of the ideas presented in Section \ref{sec:rand} and has been deferred to Appendix \ref{app:misstensors}.

\section{Conclusion}
In this paper, we study several variations on the trace reconstruction
problem to understand how structural assumptions on the input
influence the sample complexity. Our results shed light on how
sparsity, separation between \1s, randomness, and multivariate
structures can enable efficient statistical inference with the
deletion channel. Along the way, we refine existing techniques, such
as the Littlewood polynomial approach, and introduce several new
ideas, including clustering and combinatorial methods. We hope our
insights and techniques will prove useful in future work on trace
reconstruction and related problems.


\bibliographystyle{IEEEtran}
\bibliography{references}

\appendix

\section{Sparsity with gap: Technical details}
\label{app:letsgettechnical}

This section contains missing details from~\pref{sec:gaps}. Recall
that we have a string $x \in \{0,1\}^n$ that is $k$-sparse. We further
assume that each pair of successive \1s in $x$ is separated by a run of $g$ \0s, and we
refer to $g$ as the \emph{gap}.
Recall that we define $\{p_u\}_{u=1}^k$ as the position of the $k$ \1s
in original string, where $p_1 < p_2 < \ldots, p_k$. As further
notation we refer to the collection of $m=\poly(n)$ traces as $\Tc
= \{\tilde{x}_j\}_{j=1}^m$.

\paragraph*{The first level}
As a warm up, we show an algorithm called \texttt{FindPositions}, that uses
$\poly(n)$ traces to reconstruct $x$ exactly with high probability
when the gap $g = \Omega(\sqrt{n \log n})$. The algorithm returns the
values $\{p_u\}_{u=1}^{k}$ and
crucially uses a {\em binomial mean estimator}. Given $s$ samples
$X_1, X_2, \dots, X_s$ from a binomial distribution ${\rm
Bin}(n, \frac12)$ this estimator returns an estimate of $n$, $\hat{n}
= {\rm round}\Big(\frac{2}{s}\sum_{i=1}^s X_i\Big),$ where the ${\rm
round}$ function simply rounds the argument to the nearest
integer. From the Hoeffding bound, it is clear that
\begin{align*}
&\Pr(\hat{n} \ne n) = \Pr(|\hat{n} -n| \ge 1)  =  \Pr\Big(\Big|\frac{1}{s}\sum_{i=1}^s X_i - \frac{n}2\Big| \ge \frac14 \Big) \le 2 \exp\Big(-\frac{s}{8n^2}\Big) \le 2\exp(-n^\epsilon),
\end{align*} 	
as long as $s = 8n^{2+\epsilon}$ for any $\epsilon >0$. 

\begin{algorithm}
\caption{\texttt{FindPositions}}
\begin{algorithmic}
\REQUIRE length of $x$: $n$, $m$ traces $\Tc$, gap $g > 4\sqrt{2n \log (mn^3)}$. 
\STATE For each received \1, create a vertex $v$ decorated with tuple
$(z_v,t_v)$ where $z_v\in [n]$ is the position of the received \1 and
$t_v \in [m]$ is the index of the trace.
\STATE Create graph $G = (V,E)$ using vertex set above, and with edges:
\begin{align*}
E = \Big\{(v,w) : |z_v - z_w| \leq \sqrt{2n \log(mn^3)}\Big\}
\end{align*}
\STATE Find connected components $C_1,\ldots,C_{k'}$ in $G$ (If $k' \ne k$
report failure).
\STATE For each connected component $C_i$, use the binomial mean estimator on
$\{z_v\}_{v \in C_i}$ to estimate $\hat{p}_i$.
\STATE Return $\{\hat{p}_i\}_{i=1}^{k'}$. 
\end{algorithmic}
\label{alg:suff_gap}
\end{algorithm}

The algorithm \texttt{FindPositions} is displayed
in~\pref{alg:suff_gap}. Our first result of this section guarantees
that with $g = \Omega(\sqrt{n\log n})$~\pref{alg:suff_gap} recovers
$x$ exactly with $\poly(n)$ traces.
\begin{proposition}
\pref{alg:suff_gap} (\texttt{FindPositions}) successfully returns the string $x$ from $m$ traces 
with probability at least $1-3n^{-2}$ as long as
$m \geq \Omega(n^2\log n)$ and the gap $g \ge
4 \sqrt{2n \log(nm^3)}=\Theta(\sqrt{n \log n})$.
\end{proposition}
\begin{proof}
First, let us associate with each vertex $v$ an unknown label $y_v \in
[k]$ describing the correspondence between this received \1 and a \1
in the original string. The first observation is that if $y_v = u$
then $z_v \sim \textrm{Bin}(p_u, \frac{1}{2})$ and we always have
$p_u \leq n$. Thus, by Hoeffding's inequality and a union bound, we
have
\begin{align*}
&\Pr[\exists v \in V : |z_v - p_u/2| > \tau]  \leq |V|\exp(-2\tau^2/n) \leq \exp(\log(mk) - 2\tau^2/n) \ .
\end{align*}
And so with $\tau = \sqrt{n\log(mkn^{2})/2}$, with probability {\color{black} at least}
$1-n^{-2}$ all $z_v$ values concentrate appropriately.

This event immediately implies that $G$ is \emph{consistent} in the
sense that if $y_v = y_{v'}$ then $(v,v') \in E$. Further the gap
condition implies the converse property, which we call \emph{purity}:
if $y_v \ne y_{v'}$ then $(v,v') \notin E$. Formally, if $y_v \ne
y_{v'}$ then
\begin{align*}
g/2 \leq~ & |p_{y_v}/2 - p_{y_{v'}}/2| \\
\leq~ & |z_v - p_{y_v}/2| + |z_v - z_{v'}| + |p_{y_{v'}}/2 - z_{v'}| \\
\leq~ & \sqrt{2n\log(mkn^2)} + |z_v - z_{v'}|
\end{align*}
which implies that $|z_v - z_{v'}| \geq g/2
- \sqrt{2n\log(mkn^2)} > \sqrt{2n \log (mn^3)}$. Hence $(v,v') \notin
E$.

The above two properties reveal that each connected component can be
identified with a single index $u \in [k]$ corresponding to a $\1$ in the original string and the component contains
exactly the received \1s corresponding to that original one (formally
$C_u = \{v: y_v = u\}$). From here we simply use the binomial
estimator on each component. First observe that, by a Chernoff bound,
with probability {\color{black} at least} $1-k\exp(-m/36)$, each \1 from the original string
appears in at least a $1/3$-fraction of the traces, so that
$|C_u| \geq m/3$. Then apply the guarantee for the binomial mean
estimator along with another union bound over the $k$
positions. Overall the failure probability is {\color{black} at most}
\begin{align*}
n^{-2} + k\exp(-m/36) + 2k\exp\Big(\frac{-m}{24n^2}\Big)
\end{align*}
which is at most $3n^{-2}$ with $m \geq 24n^2\log(2kn^2)$. With this
choice, we can tolerate $g = O(\sqrt{n \log n})$.
\end{proof}

\paragraph*{The recursion}
The algorithm \texttt{RecurGap} (\pref{alg:recursive}) uses the clustering scheme
in~\texttt{FindPositions} in a recursive manner to estimate the
parameters $p_1, \ldots, p_k$ even when the gap $g$ is much less than
$\sqrt{n\log n}$. Define a series of threshold parameters, to be used
in each level of the recursion:
\begin{align*}
\tau_1 &= 4\sqrt{2n \log (mnk)};\\
\tau_d & = 80\sqrt{k\tau_{d-1}\log (mnk)}, \quad d=2, \ldots,D
\end{align*}
where the total number of levels is $D$. Note that, $\tau_d
\leq 80^2\cdot 4\sqrt{2}\cdot k^{1-\frac{1}{2^{d-1}}}n^{\frac{1}{2^d}}\log^{1-1/2^d}(nmk)$. In particular,
if $D= O(\log \log n)$ then we have $\tau_D = O(k\log(n))$.  

Recall that $V$ is the vertex set for the graph used above, where each
vertex $v$ corresponds to a received \1 and is associated with an unknown
original one $y_v$.  Our main result for \texttt{RecurGap} is the
following.
\begin{theorem}
\label{thm:hierarchy}
Assume $g \ge 2\tau_D$ for some $D\leq \log \log (n)$. Then with probability at least
$1-1/n$,~\pref{alg:recursive} (\texttt{RecurGap}) with $D$ levels of
recursion returns sets $S_1,\ldots,S_k \subset V$ such that {\color{black} $\forall u \in [k]$ }
\begin{enumerate}
\item $S_u \subset \{v \in V: y_v = u\}$.
\item $|S_u| \ge m/\log^5(n)$. 
\end{enumerate}
\end{theorem}

The theorem follows from the three lemmas stated earlier. Here we
restate the lemmas and provide the proofs.
\begin{lemma*}
[Consistency, restatement of~\pref{lem:consistency}]
At level $d$ let $V_{d,u} = \{v \in V_d, y_v = u\}$ for each $u \in
[k]$. Then with probability $1-1/n^2$, for each $d$ and $u$ there exists some
component $C^{(d)}_i$ at level $d$ such that $V_{d,u} \subset
C^{(d)}_i$.
\end{lemma*}

\begin{lemma*}[Length Bound, restatement of~\pref{lem:length_bd}]
At level $d$, the following holds with probability at least $1-1/n^2$:
For every component $C_i^{(d)}$ at level $d$, we have $L^{(d,i)} \leq
2k\tau_d$. Moreover if $U$ is a contiguous subsequence of
$\{1,\ldots,k\}$ with $\bigcup_{u \in U} V_{d,u} \subset C_i^{(d)}$,
then $| \min_{u \in U} p_u - \max_{u \in U} p_u | \leq 4k\tau_d$ {\color{black}with high probability}.
\end{lemma*}

\begin{lemma*}[Length Filter, restatement of~\pref{lem:filter}]
Assume $m \geq n$. At level $d$, the following holds with probability
at least $1-1/n^2$: 
For a component $C_i^{(d)}$ at level $d$, let $U$ be the maximal
contiguous subsequence of $\{1,\ldots,k\}$ such that $\bigcup_{u \in
  U} V_{d,u} \subset C_i^{(d)}$. Define $u_L = \argmin_{u \in U} p_u$
and $u_R = \argmax_{u \in U} p_u$. Then for any $v \in C_i^{(d)}$, if
$u_L$ and $u_R$ are present in $t_v$, then $v$ survives to round
$d+1$, that is $v \in V_{d+1}$. Moreover, for any $v \in V_{d+1}$, let
$p_{\min}(v,U)$ denote the original position of the first \1 from $U$
that is also in the trace $t_v$. Then we have $p_{\min}(v,U) - p_{u_L}
\leq 8\sqrt{k\tau_d\log(nmk)}$ {\color{black}with high probability}.
\end{lemma*}

The proofs of the lemmas are all-intertwined. In the induction step we
will assume that all lemmas hold at the previous level of the
recursion. Throughout we repeatedly take union bound over all $m$
traces and all up-to-$k$ components, and set the failure probability
for each event to be $1/n^2$. In applications of Hoeffding's
inequality, this produces a $2\log(nmk)$ term inside the square root.
\begin{proof}[Proof of~\pref{lem:length_bd}]
We proceed by induction. For the base case, by Hoeffding's inequality,
we know that for all $v \in V_1$ we have
\begin{align*}
|z_v - p_{y_v}/2| \leq \sqrt{n\log(mkn)} = \tau_1/8
\end{align*}
except with probability {\color{black} at most} $n^{-2}$. This means
that the position corresponding to a single index $u \in [k]$ can span
at most $\tau_1/4$ positions. {\color{black} Formally, if two vertices
  $v\ne v'$ have $y_v=y_{v'}$ then, by the triangle inequality, $|z_v
  - z_{v'}| \leq \tau_1/4$. Additionally, if two vertices $v\ne v'$
  have $y_v \ne y_{v'}$ and $|z_v - z_{v'}| \leq \tau_1/4$ (so that
  $(v,v')\in E_1$), then $|p_{y_v}/2 - p_{y_{v'}}/2| \leq
  \tau_1/2$. Use these two facts, along with the fact that there are
  at most $k$ distinct values for $y_v$, the total length of any
  connected component is at most $(k-1)\tau_1 + k\tau_1/4 \leq
  2k\tau_1$.}
The second claim follows from the concentration statement.

For the induction step, assume that the connected components at level
$d-1$ have length at most $2k\tau_{d-1}$. Fix a connected component
$C_i^{(d-1)}$ and let $u_{i,1}^{(d-1)}$ denote the left-most
original \1 present in $C_i^{(d-1)}$ ($u_{i,1}^{(d-1)} = \min\{y_v :
v \in C_{i}^{(d-1)}\}$).  By another application of Hoeffding's
inequality and using the error guarantee in~\pref{lem:filter}, we
have that
\begin{align*}
|z_v^{(d-1)} - (p_{y_v} - p_{u_{i,1}^{(d-1)}})/2| 
&\leq |z_v^{(d-1)} - (p_{y_v}-p_{\min}(v,U_i^{(d-1)}))/2| + |p_{\min}(v,U_i^{(d-1)}) - p_{u_{i,1}^{(d-1)}}|/2 \\
&\leq \sqrt{2k\tau_{d-1}\log(mkn)} + 8\sqrt{k\tau_{d-1}\log(mkn)} \leq \tau_d/8
\end{align*}
except with probability {\color{black} at most} $n^{-2}$. From here, the same argument as in
the base case yields the claim.
\end{proof}

\begin{proof}[Proof of~\pref{lem:filter}]
We have two conditions to verify. Fix a component $C_i^{(d)}$ at level
$d$ with maximal contiguous subsequence $U \subset [k]$ and recall the
definitions $u_L = \argmin_{u \in U}p_u$ and $u_R = \argmax_{u \in U}
p_u$. By another concentration bound, we know that
\begin{align*}
\forall j: \textrm{len}(\tilde{x}_j^{(d,i)}) &\leq (p_{u_R} - p_{u_L})/2  + \sqrt{(p_{u_R} - p_{u_L})\log(mnk)}
\end{align*}
with probability {\color{black} at least} $1-n^{-2}$. This reveals that:
\begin{align*}
L^{(d,i)} \leq (p_{u_R} - p_{u_L})/2 + \sqrt{(p_{u_R} - p_{u_L})\log(mnk)}
\end{align*}
Moreover, for any trace $j$ that contains $u_R, u_L$ the tail bound is
two-sided:
\begin{align*}
\forall j\ \textrm{s.t.}\ u_L,u_R \in \tilde{x}^{(d,i)}_j &: \ \Big|\textrm{len}(\tilde{x}_j^{(d,i)}) - (p_{u_R} - p_{u_L})/2\Big|  \leq \sqrt{(p_{u_R} - p_{u_L})\log(mnk)}.
\end{align*}
Note that we also have $L^{(d,i)} \ge (p_{u_R} - p_{u_L})/2$ with
overwhelming probability as:
\begin{align*}
\Pr[\forall j: \textrm{len}(\tilde{x}_j^{(d,i)}) \leq (p_{u_R}-p_{u_L})/2] \leq \prod_{j=1}^m \Pr[\len(\tilde{x}_j^{(d,i)}) 
&\leq (p_{u_R} - p_{u_L})/2 \mid u_R,u_L] \cdot \Pr[u_R,u_L]\\
& \leq \left(\frac{1}{2}\cdot\frac{1}{4}\right)^{m} = 2^{-3m}
\end{align*}
Here we are using the symmetry of the binomial distribution. Thus,
with $m \geq n$, the failure probability here is $\exp(-\Omega(n)))$,
which is negligible.

Using the upper bound on $L^{(d,i)}$ reveals that
$\tilde{x}_j^{(d,i)}$ survives, since
\begin{align*}
\textrm{len}(\tilde{x}_j^{(d,i)}) 
&\geq (p_{u_R} - p_{u_L})/2 - \sqrt{(p_{u_R} - p_{u_L})\log(mnk)}\\
& \geq L^{(d,i)} - 2\sqrt{(p_{u_R} - p_{u_L})\log(mnk)} \\
&\geq L^{(d,i)} - 2\sqrt{2L^{(d,i)}\log(mnk)}.
\end{align*}

For the second condition, assume that some trace $j$ survives but does
not contain $u_L$. Let $u_{\min} = \argmin\{y_v: v \in C_i^{(d)}, t_v
=j\}$ denote the first original \1 in this trace that belongs to
$C_i^{(d)}$s block (By definition $p_{u_{\min}}=p_{\min}(v,U)$ for each
$v: t_v = j$). Then we know that
\begin{align*}
\len(\tilde{x}_j^{(d,i)}) &\leq (p_{u_R}- p_{u_{\min}})/2 +\sqrt{ (p_{u_R} - p_{u_{\min}})\log(nmk)} \\
&\leq (p_{u_R}- p_{u_{\min}})/2 + \sqrt{2L^{(d,i)}\log(nmk)}
\end{align*}
but since $\tilde{x}_j^{(d,i)}$ passed through the length filter, we
also have a lower bound on its length, and so we get that
\begin{align*}
&p_{u_{\min}} - p_{u_L} \leq 4\sqrt{2L^{(d,i)}\log(nmk)} \leq 8\sqrt{k\tau_d\log(nmk)}
\end{align*}
where the last inequality follows from~\pref{lem:length_bd}.
\end{proof}

\begin{proof}[Proof of~\pref{lem:consistency}]
{\color{black} The proof here is similar to that of~\pref{lem:length_bd}.}
Fix
a component $C_i^{(d-1)}$ with corresponding block $U_i^{(d-1)}\subset
[k]$ at level $d-1$ and assume that all three lemmas apply for all
previous levels. For a subtrace $x_j^{(d-1,i)}$ in this component
observe and recall the definition $u_{i,1}^{(d-1)} = \min\{y_v: v \in
C_i^{(d-1)}\}$ and $p_{\min}(v,U_i^{(d-1)})$, which is the position of
the first \1 in $U_i^{(d-1)}$ that appears in trace $t_v=j$. Since the length of the subtrace is at most $2k\tau_{d-1}$ by~\pref{lem:length_bd} we get that
\begin{align}
|z_v^{(d-1)} - (p_{y_v} - p_{u_{i,1}^{(d-1)}})/2|  \notag 
&\leq |z_v^{(d-1)} - (p_{y_v} - p_{\min}(v,U_i^{(d-1)}))/2| \notag \\
&+ |p_{\min}(v,U_i^{(d-1)}) - p_{u_{i,1}^{(d-1)}}|/2 \notag \\
&\leq \sqrt{2k\tau_{d-1}\log(mnk)} + 8\sqrt{k\tau_{d-1}\log(mkn)} = \tau_d/8.\label{eq:deep_hoeffding}
\end{align}
Here the last inequality uses Hoeffding's bound along
with~\pref{lem:filter} at level $d-1$. This implies that the
clustering at level $d$ is consistent.
\end{proof}

\begin{proof}[Proof of~\pref{thm:hierarchy}]
First take a union bound over $D \leq \log \log n$ applications of the
three lemmas, so that the total failure probability is $cD/n^2 \leq
1/n$. From now, assume that the events in the three lemmas all hold
for all levels. In particular, this implies that the components
$C_i^{(D)}$ are consistent. We must verify that the clusters are pure
and then track how many vertices remain.

For the first claim, let us revisit the proof
of~\pref{lem:consistency}. If two vertices, say $v,v'$, in a
component at level $D-1$ corresponded to different \1s, say $u,u'$
then by the gap condition, we know that $|p_u - p_{u'}|\ge g$. On the
other hand, we know that~\pref{eq:deep_hoeffding} holds, and we will
use this to prove that no edge appears between these vertices. We have that
\begin{align*}
|z_v-z_{v'}| &\geq |p_{y_v} - p_{y_{v'}}|/2 - \tau_{D}/8 - \tau_{D}/8 \\
&\ge g/2 - \tau_D/4,
\end{align*}
and so, if $g/2 \geq \tau_D$, then the two vertices will not share an
edge. The argument applies for all pairs and hence the clusters at
level $D$ are pure, which establishes the first claim in
the~\pref{thm:hierarchy}.

For the second claim, note that by~\pref{lem:filter}, for every
component at every level, if a trace contains the two endpoints of
that component, then it will survive the filter. Hence, in every
filtering step we expect to retain $1/4$ of the subtraces passing
through, and, by a Chernoff bound, we will retain $1/5$ of the
subtraces except with $\exp(-\Omega(n))$, provided $m \geq n$. Since
we perform $D=\log\log n$ levels, we retain $m/5^{\log \log n} =
m/\log^5(n)$ traces in each cluster with high probability.
\end{proof}

\paragraph*{Removing Bias: The reverse recursion}

Now that we have isolated the vertices into pure clusters, we need to
work our way up through the recursion to remove biases introduced by
the hierarchical clustering. For any component $C_i^{(D-1)}$
corresponding to block $U_i^{(D-1)} \subset [k]$ at level $D-1$, since
the components at level $D$ are pure, we can identify exactly the
subtraces that contain the first and last \1 in the block. We throw
away all other traces, which de-biases the length filter at level
$D-1$.

Unfortunately for a component $C_i^{(d-1)}$ corresponding to a block
$U_i^{(d-1)}$ at level $d-1$, we cannot identify exactly the subtraces
that contain the exactly the first and last \1 in the block. However,
we know that $C_i^{(d-1)}$ is further refined into sub-components
$\{C_{i'}^{(d)}\}$ at level $d$, and by induction we can identify all
the traces that contain the left-most and right-most \1 in the
left-most and right-most sub-components. We identify all such traces
and eliminate the rest to debias the length filter at level
$d-1$. See~\pref{fig:debias} for an illustration.

To debias this length filter, we filter based on the presence of
two \1s at level $d-1$ (just the end points), and two futher \1s at
level $d$ (the inner endpoints of the first and last sub-components),
four further \1s at $d+1$, and so on. So, just to debias the length
filter at level $d-1$ we require $2^{D-(d-1)}$ \1s to be
present. Since we must debias all length filters above a particular
component, we require the presence of $\sum_{d=1}^{D-1} 2^{D-d} \leq
2^{D} \leq \log_2(n)$ \1s. The probability of all $\log_2(n)$ of
these \1s appearing is $1/n$ and by Chernoff bound, with high
probability at least $m/2n$ of our traces will contain all of
these \1s.

For any \1, $u$, in the original string, let $S$ denote the subset of
$\log_2(n)$ \1s, whose presence we require to debias the length
filters above the pure component containing $u$. After the debiasing
step, the remaining vertices in the component containing $u$ have
$z_v$ values distributed as
\begin{align*}
z_v \sim \textrm{Bin}(p_u - 1 - |S_L|, 1/2) + (|S_L|+1)
\end{align*}
where $|S_L|$ is the number of \1s in $|S|$ that appear before $u$ in
the sequence, and the final 1 is due to the presence of $u$. Using the
binomial mean estimator, we can therefore estimate $p_u$ with
probability {\color{black} at least} $1-O(1/n)$, provided $m \ge n^2\log(n)$. Thus, $\poly(n)$
traces suffice to recover all $p_u$ values, provided that $g > \tau_D$
and $D = \log_2 \log_2 n$. This proves~\pref{thm:gaps_intro}.

\section{Missing Proofs from Section \ref{sec:tensors}}
\label{app:misstensors}

\begin{proof}[Proof of Lemma \ref{lem:multivariate_littlewood}]
Fix $L > 0$ and define the polynomial
\begin{align*}
&F(z_1,z_2,\dots,z_k) =\prod_{1\leq a_1,a_2,\dots,a_k \leq L} f( z_1 e^{\frac{\pi i a_1 }{L}}, z_2 e^{\frac{\pi i a_2 }{L}},\dots,z_k e^{\frac{\pi i a_k }{L}} ).
\end{align*}
We first show that there exists $z^\star_1,z^\star_2,\dots,z^\star_k$ on the unit disk ($|z^{\star}_1|=|z^{\star}_2|=\dots=|z^{\star}_k|=1$) such that
$F(z^\star_1,z^\star_2,\dots,z^\star_k) \geq 1$. This follows from an iterated application of the
maximum modulus principle. First factorize $F(z_1,z_2,\dots,z_k) = z_k^{s_k}
F^1(z_1,z_2,\dots,z_k)$ where $s_k$ is chosen such that $F^1(z_1,z_2,\dots,z_k)$ has no common
factors of $z_k$. Since $F$ has non-zero coefficients, this implies
that $F^1(z_1,z_2,\dots,0)$ is a non-zero polynomial and therefore using the maximum modulus principle, for any fixed $z_1,z_2,\dots,z_{k-1}$, there exists a value of $z_k=\bar{z}_k$ such that $|\bar{z}_k|=1$ and 
\begin{align*}
|F(z_1,z_2,\dots,\bar{z}_k)| \ge |F^1(z_1,z_2,\dots,0)|.
\end{align*}
Subsequently we can further factorize $F^1(z_1,z_2,\dots,0)=z_{k-1}^{s_{k-1}}F^2(z_1,z_2,\dots,z_{k-1})$ so that $F_2(z_1,z_2,\dots,z_{k-1})$ has no common factors in $z_{k-1}$. Repeating this procedure $k$ times, we can show the following chain of inequalities 
\begin{align*}
|F(z^\star_1,z^\star_2,\dots,z^{\star}_k)| &= |F^1(z^\star_1,z^\star_2,\dots,z^{\star}_k)| \geq |F^1(z^\star_1,\dots,0)| \\
 &= |F^2(z^\star_1,z^\star_2,\dots,z^{\star}_{k-1})| \\
 &\geq |F^2(z^\star_1,\dots,0)| \geq \dots |F^k(z^{\star}_1)|\ge |F^k(0)| \ge 1 
\end{align*}

Now, for any $a_1,a_2,\dots,a_k \in \{1,\ldots,L\}$ we have
\begin{align*}
1 \leq |F(z_1^\star,z_2^\star,\dots,z_k^{\star})| \leq |f(z_1^\star e^{\pi ia_1/L}, z_2^\star e^{\pi i a_2/L},\dots, z_k^\star e^{\pi i a_k/L})| \cdot n^{(L^k-1)},
\end{align*}
where we are using the fact that $|f(z_1,z_2,\dots,z_k)| \leq n$. This proves the lemma, since we may choose $a_1,a_2,\dots,a_k$ such that $z_j^\star e^{\pi ia_j /L} = \exp(i\theta_j)$ for $|\theta_j| \leq \pi/L$ for all $j=1,2,\dots,k$. 
\end{proof}

\begin{proof}[Proof of Lemma \ref{lem:distincttensors}]
For $k$ complex numbers $w_1,w_2,\dots,w_k \in \mathbb{C}$, observe that
\begin{align*}
\mathbb{E}\left[\sum_{i_1,i_2,\dots,i_k=0}^{n^{1/k}-1}\tilde{T}_{i_1,i_2,\dots,i_k}w_1^{i_1}w_2^{i_2}\dots w_k^{i_k}\right]  
&= q^k\sum_{i_1,i_2,\dots,i_k}w_1^{i_1}w_2^{i_2}\dots w_k^{i_k}  \sum_{j_1 \ge i_1,j_2 \ge i_2,\dots,j_k \ge i_k} X_{j_1,j_2,\dots,j_k} \prod_{t=1}^{k} {j_t \choose i_t}p^{j_t-i_t}q^{i_t} \\
& = q^k\sum_{j_1,j_2,\dots,j_k}^{n^{1/k}-1} X_{j_1,j_2,\dots,j_k} \prod_{t=1}^{k}(qw_t + p)^{j_t}
\end{align*}
Thus, for two tensors $T_1,T_2$, we have
\begin{align*}
\frac{1}{q^k} \mathbb{E}\Big[\sum_{i_1,i_2,\dots,i_k=0}^{n^{\frac{1}{k}}-1}(\tilde{T_1}_{i_1,i_2,\dots,i_k} -\tilde{T_2}_{i_1,i_2,\dots,i_k})w_1^{i_1}w_2^{i_2}\dots w_k^{i_k}\Big] 
&  = \sum_{j_1,j_2,\dots,j_k=0}^{n^{1/k}-1} \Big({T_1}_{j_1,j_2,\dots,j_k}-{T_1}_{j_1,j_2,\dots,j_k}\Big)  \prod_{t=1}^{k}(qw_t + p)^{j_t}\\
& \triangleq A(z_1,z_2,\dots,z_k)
\end{align*}
where we are rebinding $z_t = qw_t+p$ for all $t=1,2,\dots,k$. Observe that
$A(z_1,z_2,\dots,z_k)$ is a \emph{multivariate Littlewood polynomial}; all
coefficients are in $\{-1,0,1\}$, and the degree is $n^{1/k}$ {\color{black}in each variable}.

Again, for $z_1,z_2,\dots,z_k \in \gamma_L \equiv \{e^{i\theta}: |\theta| \leq \pi/L\}$ we can use Lemma \ref{lem:multivariate_littlewood} and the fact that
\begin{align*}
\forall z \in \gamma_L: |(z-p)/q| \leq \exp(C_2p/(Lq)^2).
\end{align*}
 to sandwich $|A(z_1,z_2,z_3,\dots,z_k)|$ by
\begin{align*}
\exp(-C_1L^k\log n) &\leq \max_{z_1,z_2,\dots,z_k \in \gamma_L} |A(z_1,z_2,\dots,z_k)| \\
& \leq \frac{\exp(Ckpn^{1/k}/(Lq)^2)}{q^k} \sum_{ij}\left|\EE[\tilde{T_1}_{i_1,i_2,i_k} - \tilde{T_2}_{i_1,i_2,\dots,i_k}]\right|  \ . 
\end{align*}
This implies that there exists $i_1,i_2,\dots,i_k$ such that
\begin{align*}
\left|\EE[\tilde{T_1}_{i_1,i_2,\dots,i_k} - \tilde{T_2}_{i_1,i_2,\dots,i_k}]\right| 
&\geq \frac{q^k}{n} \exp\left(-C_1L^k\log n - \frac{Ckpn^{1/k}}{L^2q^2}\right) \\
& \geq \frac{q^k}{n} \exp\left(-C n^{1/(k+2)} \Big( \frac{kp \log^{2/k} n}{q^2} \Big)^{k/(k+2)}\right),
\end{align*}
where the second inequality follows by optimizing for $L$. 
\end{proof}

\begin{proof}[Proof of Theorem \ref{thm:randomtensors}]
We will use the oracle described in Section \ref{sec:rand} again. Recall that the oracle was able to distinguish between the following two cases 
\begin{description}
\item[Case 1:]~ $t$ and $t'$ are traces generated by the deletion channel with preservation probability $q=1-p$ from the same random string $x\in_R \{0,1\}^{n^{1/k}}$
\item[Case 2:]~ $t$ and $t'$ are traces generated by the deletion channel with preservation probability $q=1-p$ from independent random strings $x,y\in_R \{0,1\}^{n^{1/k}}$
\end{description}
with failure probability at most $1/n^{20/k}$. 

Notice that the probability of a particular bit in $T$ getting deleted is $1-q^k$. In that case, with $m=2\log n/q^k$ traces we can ensure that every bit of $X$ appears in at least one of the tensor traces with probability at least $1-\frac{1}{n}$. Suppose we fix $k-1$ dimensions and without loss of generality suppose we fix the value of the $r^{th}$ dimension of $T$ to be $i_r$ for all $r\neq 1$. In that case the elements $\{T_{j,i_2,i_3,\dots,i_k}\}_{j=1}^{n^{1/k}}$ form a binary vector of length $\{0,1\}^{n^{1/k}}$. There are $n^{(k-1)/k}$ such binary vectors corresponding to the $n^{(k-1)/k}$ different values of $i_2,i_3,\dots,i_k$ and we will denote the set of traces from the $l^{th}$ such binary vector by $G_l$. Notice that there exists a natural ordering among these groups $\{G_l\}_{l=1}^{n^{(k-1)/k}}$. For two distinct groups $G_l,G_{l'}$, where $l,l'$ is defined by $(i_2,i_3,\dots,i_k)$ and $(j_2,j_3,\dots,j_k)$ respectively,  we will have $l<l'$ if and only if there exists a value $r\le k$ such that 
\begin{align*}
i_r<j_r \quad \text{ and } \quad i_t \le j_t \quad \forall t\in \{r+1,r+2,\dots,k\}.
\end{align*} 
Moreover, when we observe a tensor trace after fixing all the dimensions, except the first one, we actually observe the vector traces of one of those $n^{(k-1)/k}$ binary vectors. Suppose for every tensor trace, we do this process and collect all the vector traces by fixing every dimension except the first one. We can now use our oracle to group all these vector traces according to the original binary vector they emanated from i.e two vector traces belong to the same group if both of them belong to $G_l$ for some value of $l \in [n^{(k-1)/k}]$. This requires at most $\binom{m n^{1/k}}{2}\leq (mn^{1/k})^2$ applications of the oracle and so, by the union bound, this can performed with failure probability  at most $$\frac{(mn^{1/k})^2}{n^{20/k}} \leq \frac{2^{2k+2}\log^2 n}{n^{18/k}}$$ where the inequality applies for sufficiently large $n$. We next infer the ordering among the $n^{(k-1)/k}$ groups  $\{G_l\}_{l=1}^{n^{(k-1)/k}}$. For two distinct $l,l' \in  [n^{(k-1)/k}]$, where $l,l'$ is defined by $(i_2,i_3,\dots,i_k)$ and $(j_1,j_2,\dots,j_k)$ respectively, suppose there exists a tensor trace having at least one vector trace from both $G_l$ and $G_{l'}$. Moreover suppose the position of the vector trace from $G_l$ is given by $(\tilde{i}_2,\tilde{i}_3,\dots,\tilde{i}_k)$ and the position of the vector trace from $G_{l'}$ is given by $(\tilde{j}_2,\tilde{j}_3,\dots,\tilde{j}_k)$. In that case, we will infer that $l<l'$ if there exists an $r\le k$ such that 
\begin{align*}
\tilde{i}_r< \tilde{j}_r \quad \text{ and } \quad \tilde{i}_t \le \tilde{j}_t \quad \forall t\in \{r+1,r+2,\dots,k\}.
\end{align*} 
and infer $l>l'$ otherwise. 
The probability there exists such a trace is $1-(1-q^2)^m\geq 1-1/\poly(n)$. 
We also perform an analogous process with every such dimension. After all dimensions have been processed, we know exactly the elements along each dimension that has been deleted to form each tensor trace, which subsequently reveals the original position of each received bit in each tensor trace. Given that every bit of $X$ appeared in at least some trace, this suffices to reconstruct $X$, proving the main theorem.
\end{proof}

\end{document}